\documentclass[11pt]{article}

\usepackage[margin=0.8in]{geometry}

\usepackage{varioref}
\usepackage[usenames,svgnames,xcdraw,table]{xcolor}
\definecolor{DarkBlue}{rgb}{0.1,0.1,0.5}
\definecolor{DarkGreen}{rgb}{0.1,0.5,0.1}

\usepackage{nicefrac}

\usepackage{hyperref}
\hypersetup{
   colorlinks   = true,
 linkcolor    = DarkBlue, 
 urlcolor     = DarkBlue, 
	 citecolor    = DarkGreen 
}

\usepackage{appendix}

\usepackage{mathtools}

\usepackage{amsthm}
\usepackage{algorithmic}
\usepackage{algorithm}
\usepackage{array}

\usepackage[font={small,it}]{caption}

\usepackage{graphicx,epsfig,amsmath,latexsym,amssymb,verbatim}

\usepackage[square,semicolon,numbers]{natbib}

\newcommand{\extra}[1]{}

\usepackage{amsmath,amssymb,amsfonts}
\usepackage{amsthm}
\usepackage{thmtools,thm-restate}
\usepackage{enumitem}

\newtheorem{theorem}{Theorem}

\newtheorem{definition}{Definition}
\newtheorem{lemma}[theorem]{Lemma}
\newtheorem{proposition}[theorem]{Proposition}

\def\squareforqed{\hbox{\rlap{$\sqcap$}$\sqcup$}}
\def\qed{\ifmmode\squareforqed\else{\unskip\nobreak\hfil
\penalty50\hskip1em\null\nobreak\hfil\squareforqed
\parfillskip=0pt\finalhyphendemerits=0\endgraf}\fi}
\def\endenv{\ifmmode\;\else{\unskip\nobreak\hfil
\penalty50\hskip1em\null\nobreak\hfil\;
\parfillskip=0pt\finalhyphendemerits=0\endgraf}\fi}

\renewenvironment{proof}{\noindent \textbf{{Proof~} }}{\qed\medskip}
\newenvironment{proof+}[1]{\noindent \textbf{{Proof #1~} }}{\qed\medskip}

\mathchardef\ordinarycolon\mathcode`\:
\mathcode`\:=\string"8000
\def\vcentcolon{\mathrel{\mathop\ordinarycolon}}
\begingroup \catcode`\:=\active
  \lowercase{\endgroup
  \let :\vcentcolon
  }

\DeclareMathOperator*{\argmin}{arg\,min}

\usepackage{amsmath}
\usepackage{amssymb,amsthm,mathtools}

\usepackage{amsmath,amsfonts} 
\usepackage{caption} 
\usepackage{xcolor}
\usepackage{hyperref}
\usepackage{comment}

\usepackage{verbatim}

\usepackage[font=small,labelfont=bf]{caption}
\usepackage{xspace}
\usepackage{etoolbox}

\title{\bfseries Approximating Nash Social Welfare Under Binary $\XOS$ and Binary Subadditive Valuations}

\author{Siddharth Barman\thanks{Indian Institute of Science. {\tt barman@iisc.ac.in}} \and Paritosh Verma\thanks{Purdue University. {\tt paritoshverma97@gmail.com}}}

\date{}

\newcommand{\XOS}{\mathrm{XOS}}

\newcommand{\alloc}{\mathcal{A}}
\newcommand{\palloc}{\mathcal{P}}
\newcommand{\SW}{\operatorname{SW}}

\newcommand{\const}{18}
\newcommand{\constMinTwo}{16}
\newcommand{\approxConst}{288}              

\newcommand{\NSW}{\mathrm{NSW}}
\newcommand{\MMS}{\mathrm{MMS}}
\newcommand{\GMMS}{\mathrm{GMMS}}

\newcommand{\Alg}{\textsc{Alg}}           
\newcommand{\breakone}{p}
\newcommand{\breaktwo}{q}
             
\begin{document}

\maketitle 

\begin{abstract}
We study the problem of allocating indivisible goods among agents in a fair and economically efficient manner.  In this context, the Nash social welfare---defined as the geometric mean of agents' valuations for their assigned bundles---stands as a fundamental measure that quantifies the extent of fairness of an allocation. Focusing on instances in which the agents' valuations have binary marginals, we develop essentially tight results for (approximately) maximizing Nash social welfare under two of the most general classes of complement-free valuations, i.e., under binary $\XOS$ and binary subadditive valuations.  

For binary $\XOS$ valuations, we develop a polynomial-time algorithm that finds a constant-factor (specifically $\approxConst$) approximation for the optimal Nash social welfare, in the standard value-oracle model. The allocations computed by our algorithm also achieve constant-factor approximation for social welfare and the groupwise maximin share guarantee. These results imply that---in the case of binary $\XOS$ valuations---there necessarily exists an allocation that simultaneously satisfies multiple (approximate) fairness and efficiency criteria. We complement the algorithmic result by proving that Nash social welfare maximization is $\mathrm{APX}$-hard under binary $\XOS$ valuations.

Furthermore, this work establishes an interesting separation between the binary $\XOS$ and binary subadditive settings. In particular, we prove that an exponential number of value queries are necessarily required to obtain even a sub-linear approximation for Nash social welfare under binary subadditive valuations. 
\end{abstract}

\section{Introduction}
At the core of discrete fair division lies the problem of fairly allocating \emph{indivisible} goods among agents with equal entitlements, but distinct preferences. In this context, the Nash social welfare \cite{Nash50}---defined as the geometric mean of agents' valuations for their assigned bundles---stands as a fundamental measure that quantifies the extent of fairness of an allocation. This welfare function achieves a balance between the extremes of social welfare and egalitarian welfare. The relevance of Nash social welfare is further substantiated by the fact that it satisfies key fairness axioms, including the Pigou-Dalton transfer principle \cite{Moul03}. Furthermore, Nash social welfare is indifferent to individual scales of the valuations: multiplicatively scaling any agent's valuation by a positive number does not alter the relative ordering of the allocations (induced by this welfare objective) and, in particular, keeps the Nash optimal allocation unchanged. In terms of practical applications, Nash social welfare is used as an optimization criterion by the widely-used platform {\tt Spliddit.org} for finding fair allocations~\cite{GP15}. 

With these considerations in hand, a substantial body of research in recent years has been directed towards maximizing the Nash social welfare in settings with indivisible goods; see, e.g., \cite{CG15,CDGJ+17,AGS+17nash,BGH+17earning,AGM+18nash,BKV18,GHM18}. This maximization problem is known to be $\mathrm{APX}$-hard, even when the agents have additive valuations \cite{Lee17}.\footnote{Recall that a valuation $v$ is said to be additive iff the value of any subset of goods, $S$, is the sum of values of the goods in it, $v(S) = \sum_{g \in S} v(\{g\})$.} Hence, in general, algorithmic results for this problem aim for approximation guarantees. A key focus of this line of research has been on the hierarchy of complement-free valuations, which includes the following valuation classes, in order of containment: additive, submodular, $\XOS$ (fractionally subadditive), and subadditive. Recall that submodular functions satisfy a diminishing returns property, $\XOS$ functions are pointwise maximizers of additive functions, and subadditive functions constitute the most general class in this hierarchy. These valuation classes have also been extensively studied in the literature on combinatorial auctions, wherein the focus is primarily on maximizing social welfare \cite{nisan2007algorithmic}. 

In the context of maximizing Nash social welfare, the best-known approximation algorithm for additive valuations achieves an approximation ratio of $e^{1/e}$ (in polynomial time) \cite{BKV18}. Furthermore, for submodular valuations, a recent result of  Li and Vondr\'{a}k \cite{LV21} obtains a constant-factor (specifically $380$) approximation ratio; see also \cite{garg2021approximating}. In contrast to these constant-factor bounds, the problem of maximizing Nash social welfare under (general) $\XOS$ and subadditive valuations has a linear (in the number of agents) approximation guarantee \cite{BBKS20, chaudhury2020fair,GargKK20}. This approximation ratio is in fact tight in the standard value-oracle model: under general $\XOS$ (and, hence, subadditive) valuations, a sub-linear approximation of the optimal Nash social welfare necessarily requires exponentially many value queries \cite{BBKS20}. 

The current work contributes to this thread of research with a focus on valuations that have binary marginals. Formally, a valuation $v$ is said to bear the binary-marginals property iff, for every subset of goods $S$ and each good $g$, the marginal value of $g$ relative to $S$ is either zero or one, $v( S  \cup \{g \}) - v(S) \in \{0,1\}$. Such valuations capture preferences in many real-world domains and have received significant attention in the fair division literature; see, e.g., \cite{BMS05,RSU05,BL08,F10,KPS18,O20}. Our results address the two most general valuation classes in the above-mentioned hierarchy, i.e., we study $\XOS$ and subadditive valuations in conjunction with the binary-marginals property. This meaningfully extends prior work on binary additive and binary submodular valuations. Throughout, we will say that a valuation is binary additive (submodular/$\XOS$/subadditive) iff it is additive (submodular/$\XOS$/subadditive) and also has binary marginals. 

In particular, under binary additive valuations,  Nash optimal allocations can be computed in polynomial time~\cite{DS15,barman2018greedy}. The work of Halpern et al.~\cite{halpern2020fair} shows that---in the case of binary additive valuations---maximizing Nash social welfare (with a lexicographic tie-breaking rule)  provides a truthful and fair\footnote{Specifically, the mechanism of Halpern et al.~\cite{halpern2020fair} ensures envy-freeness up to one good.} mechanism. For the broader class of binary submodular valuations,\footnote{Binary submodular functions admit the following characterization: every binary submodular function is necessarily a rank function of a matroid \cite[Chapter~39]{schrijver2003combinatorial}.} a truthful, fair, and polynomial-time  mechanism was obtained by Babaioff et al.~\cite{BEF2020}; this result considers {\it Lorenz domination} as a notion of fairness and, hence, ensures that the computed allocation maximizes the Nash social welfare and satisfies other fairness criteria. These works identify multiple domains wherein binary additive and binary submodular functions are applicable; see also \cite{benabbou2020finding}. 

The current work moves up in the hierarchy of complement-free valuations and develops essentially tight results for both binary $\XOS$ and binary subadditve valuations. Before detailing our results, we provide a stylized example that illustrates the relevance of such a generalization: consider a spectrum-allocation setting wherein transmission rights of distinct (frequency) bands have to be fairly allocated among different agents. Here, each band is an indivisible good of unit value, to every agent. However, an agent can utilize a subset of bands only if the frequencies across the allocated bands are close enough. In particular, say the bands, $B_1, B_2,\ldots, B_m$, are indexed such that an agent can use a subset of bands only if their indices are within a parameter $\Delta  \in \mathbb{Z}_+$ of each other. For instance, if $\Delta  = 3$, then an agent will have value two for the bundle $\{B_1, B_4, B_{10}\}$ and value the bundle $\{B_1, B_4, B_5, B_7\}$ at three, by transmitting on $B_4$, $B_5$, and $B_7$. We note that such valuations can be expressed as binary $\XOS$ functions (and not as submodular functions).  Hence, in such a resource-allocation setting, finding a fair, or economically efficient, allocation falls under the purview of the current work. For a realistic treatment of spectrum allocations, and associated high-stakes auctions, see \cite{LMS17} and references therein. 

\paragraph{Our Results.}
The following list summarizes our contributions. We note that the algorithmic results in this work only require access to standard value queries: given any subset of goods $S$ and an agent $i$, the value oracle returns the value that $i$ has for $S$.\footnote{That is, the developed algorithms do not require an explicit description of the valuations. Note that in the current context the agents' valuations are combinatorial set functions, hence explicitly  representing the valuations might be prohibitive, i.e., require one to specify exponential (in the number of goods) values.}  

\begin{enumerate}[leftmargin=0.4cm]
\item We develop a polynomial-time $\approxConst$-approximation algorithm for maximizing Nash social welfare under binary $\XOS$ valuations (Theorem \ref{theorem:approximation-ratio}). We also complement this algorithmic result by proving that---under binary $\XOS$ valuations---Nash social welfare maximization is $\mathrm{APX}$-hard (Theorem \ref{theorem:apx-hard}). 

To obtain the approximation guarantee, we consider allocations wherein, for each agent $i$, the envy is multiplicatively bounded towards the entire set of goods, $G_i$, allocated to agents with bundle size at least four times that of $i$. Specifically, for an allocation, write $H_i$ to denote the set of agents who have received a bundle of size at least four times that of $i$, and let $G_i$ denote the set of goods allocated among all the agents in $H_i$ (along with unallocated goods, if any). We show that, under binary $\XOS$ valuations, if, in an allocation $\alloc$, each agent $i$'s value for her own bundle is at least $1/2$ times her value for $G_i$, then $\alloc$ achieves a constant-factor approximation guarantee for Nash social welfare. Our algorithm (Algorithm \ref{algo:apx-mnw}) finds such an allocation by iteratively updating the agents' bundles towards the desired property. The algorithm also maintains an analytically useful property that each agent's value for her bundle is equal to the cardinality of the bundle, i.e., the bundles are \emph{non-wasteful}. For binary $\XOS$ valuations, one can show that multiplicatively bounding envy between pairs of agents does not, by itself, provide a constant-factor approximation guarantee (see Appendix \ref{appendix:nomas-envy}). Hence, bounding envy of every agent $i$ against all of $G_i$ is a crucial extension. It is relevant to observe that while the algorithm is simple, its analysis is based on novel counting arguments (Lemma \ref{lemma:geometric-relationship} and Proposition \ref{prop:constrained-packing}). Notably, the combinatorial nature of the algorithm makes it amenable to large-scale implementations. 

\item Furthermore, our algorithm (Algorithm \ref{algo:apx-mnw}) achieves an approximation ratio of $(3+2 \sqrt{2})$ for the problem of maximizing social welfare under binary $\XOS$ valuations (Theorem \ref{theorem:sw}). That is, the computed allocation simultaneously provides approximation guarantees for Nash social welfare (a fairness metric) and social welfare (a measure of economic efficiency). 

In addition, the allocation (approximately)  satisfies the fairness notion of \emph{groupwise maximin shares} ($\GMMS$); see Section \ref{section:gmms} for definitions. $\GMMS$ is a stronger criterion than the well-studied fairness concept of maximin shares ($\MMS$). Specifically, an allocation $\alloc$ is said to be $\alpha$-$\GMMS$ iff $\alloc$ is $\alpha$-approximately $\MMS$ for every subgroup of agents. The allocations computed by our algorithm are $\nicefrac{1}{6}$-$\GMMS$ (Theorem \ref{theorem:gmms}).   

\item Complementing the above-mentioned positive results, we prove that, under binary subadditive valuations, an exponential number of value queries are necessarily required to obtain a sub-linear approximation for the Nash social welfare (Theorem \ref{theorem:query-complexity}). Indeed, this query complexity bound identifies an interesting dichotomy between the binary subadditive and the binary $\XOS$ settings: while for binary $\XOS$ valuations a polynomial-number of value queries suffice for approximating the optimal Nash social welfare within a constant factor, binary subadditive valuations are essentially as hard as general subadditive (or general $\XOS$) valuations.  
\end{enumerate}

\paragraph{Additional Related Work.} The current work provides a single algorithm that achieves constant-factor approximation guarantees for both Nash social welfare and social welfare, under binary $\XOS$ valuations. Focusing solely on social welfare maximization, one can compute an $\left(e/(e-1) \right)$-approximation (of the optimal social welfare) under general $\XOS$ valuations, using the algorithm of Feige \cite{Feige06}. This result, however, requires oracle access to \emph{demand queries}, which is a more stringent requirement than one used in the current work (of value queries).\footnote{The question of whether sub-linear approximation bounds can be achieved for Nash social welfare with demand-oracle access to general $\XOS$, and subadditive, valuations remains an interesting direction of future work.} 

We obtain the query complexity, under  binary subadditive valuations, by utilizing a lower-bound framework of Dobzinski et al.~\cite{DobzinskiNS10}. In~\cite{DobzinskiNS10}, a lower bound was obtained---for social welfare maximization---under general $\XOS$ and subadditive valuations. The notable technical contribution of the current work is to establish the query complexity with valuations that in fact have binary marginals. Furthermore, one can show that our lower bound (Theorem \ref{theorem:query-complexity}) holds more broadly for maximizing \emph{$p$-mean welfare}, for any $p \leq 1$; this includes social welfare maximization as a special case. Hence, we also strengthen the negative result of \cite{BBKS20}, by showing that it continues to hold even if the marginals of the subadditive valuations are binary. 

Maximin share ($\MMS$) is a prominent fairness notion in discrete fair division \cite{Bud11}. For an agent $i$, this fairness threshold is defined as the maximum value that $i$ can guarantee for herself by partitioning the set of goods into $n$ bundles and receiving the minimum valued one; here, $n$ denotes the total number of agents. While $\MMS$ allocations (i.e., allocations that provide each agent a bundle of value at least as much as her maximin share) are not guaranteed to exist \cite{PW14,KPW16}, 
this fairness notion is quite amenable to approximation guarantees across the hierarchy of complement-free valuations; see, \cite{GT20}, \cite{ghodsi2018fair}, and references therein. Specifically, in the binary-marginals case, $\MMS$ allocations are guaranteed to exist and can be computed efficiently for binary additive \cite{BL16} and binary submodular \cite{BV21} valuations. By contrast, such an existential result does not hold for binary $\XOS$ valuations \cite{BV21}. For such valuations, however, the work of Li and Vetta \cite{li2021fair} provides a polynomial-time algorithm that finds allocations wherein each agent receives a bundle of value at least $0.367$ times her maximin share.\footnote{The result of Li and Vetta \cite{li2021fair} holds for a somewhat more general valuation class, which are  defined via \emph{hereditary set systems}.} The current work addresses the stronger notion of groupwise maximin shares \cite{BBKN18} under binary $\XOS$ valuations.

\section{Notation and Preliminaries}
\label{section:prelims}
We study the problem of allocating $m$ indivisible goods among $n$ agents in a fair and economically efficient manner. Throughout, we will use $[m] \coloneqq \{1,2, \ldots ,m\}$ to denote the set of goods and $[n] \coloneqq \{1,2, \ldots, n\}$ to denote the set of agents. The cardinal preference of the agents $i \in [n]$, over subsets of goods, are expressed via valuations $v_i : 2^{[m]} \mapsto \mathbb{R}_+$. Specifically, $v_i(S) \in \mathbb{R}_+$ denotes the value that agent $i \in [n]$ has for subset of goods $S \subseteq [m]$. We represent fair division instances by the triple $\langle [m], [n], \{v_i\}_{i=1}^n \rangle$. 

An allocation $\alloc = (A_1, A_2, \ldots, A_n)$ is a collection of $n$ pairwise disjoint subsets of goods, $A_i \cap A_j = \emptyset$ for all $i \neq j$. Here, the subset of goods $A_i  \subseteq [m]$ is assigned to agent $i \in [n]$ and will be referred to as a bundle. For ease of presentation and analysis, we do not force the requirement that, in an allocation, all the goods are assigned, i.e., the allocations can be partial with $\cup_{i=1}^n A_i \neq [m]$. Write $A_0 \coloneqq [m] \setminus \left( \cup_{i=1}^n A_i\right)$ to denote the subset of unassigned goods in an allocation $\alloc =  (A_1,\ldots, A_n)$.\footnote{Note that one can always allocate the subset of unassigned goods $A_0$ arbitrarily among the agents without reducing the Nash social welfare of $\alloc =  (A_1, A_2, \ldots, A_n)$.}  

\paragraph{Valuation Classes.} This work focuses on valuations that have the binary-marginals property, i.e., are {dichotomous}. Formally, a valuation $v$ is said to have {binary marginals} iff $v (S \cup \{ g \}) - v (S) \in \{0,1\}$ for any subset of goods $S \subseteq [m]$ and any good $g \in [m]$. As a direct consequence, the valuations we consider are monotonic: $v(S) \leq v(T)$ for any subsets $S \subseteq T \subseteq [m]$. In addition, we assume that the valuations are normalized, $v_i(\emptyset) = 0$ for each $i \in [n]$. 

A set of goods $S \subseteq [m]$ is said to be \emph{non-wasteful}, with respect to a valuation $v$, iff $v(S) = |S|$. Note that, under valuations with binary marginals, subsets of non-wasteful sets are also non-wasteful; a proof of the following proposition is provided in Appendix \ref{appendix:prelims}.

\begin{restatable}{proposition}{PropNonWasteful}
\label{proposition:non-wasteful-subsets}
Let $v: 2^{[m]} \mapsto \mathbb{Z}_+$ be a function with binary marginals and $S \subseteq [m]$ be a non-wasteful set (with respect to $v$), then each subset of $S$ is non-wasteful as well.   
\end{restatable}

We consider valuations that---in conjunction with satisfying the binary-marginals property---belong to the following classes of complement-free functions, presented in order of containment. \\

\noindent
(i) \emph{Additive}: A valuation $v: 2^{[m]} \mapsto \mathbb{R}_+$ is said to be additive iff the value of any subset of goods $S \subseteq [m]$ is equal to the sum of values of the goods in it,  $v(S) = \sum_{g \in S} v(\{g\})$. \\

\noindent
(ii) \emph{Submodular}: A valuation $v: 2^{[m]} \mapsto \mathbb{R}_+$ is  said to be submodular iff $v(S \cup \{g\}) - v(S) \geq v(T \cup \{ g \}) - v(T)$ for every $S \subseteq T \subset [m]$ and good $g \in [m] \setminus T$. \\

\noindent
(iii) \emph{$\XOS$}: A valuation $v: 2^{[m]} \mapsto \mathbb{R}_+$ is said to be $\XOS$ iff it can be expressed as a pointwise maximum over a collection of additive functions, i.e., there exists a collection of additive functions $\{\ell_t\}_{t=1}^L$ such that, for every subset $S \subseteq [m]$, we have $v(S) = \max_{ t \in [L] } \ \ell_t(S)$. Here, the number of additive functions, $L$, can be exponentially large in $m$. \\

\noindent
(iv) \emph{Subadditive}: A valuation $v: 2^{[m]} \mapsto \mathbb{R}_+$ is said to be subadditive iff it does not admit any complementary subset of goods: $v(S \cup T) \leq v(S) + v(T)$, for every pair of subsets $S,T \subseteq [m]$. \\

As mentioned previously, our algorithmic results hold in the standard value-oracle model, wherein, given any subset of goods $S \subseteq [m]$ and an agent $i \in [n]$, the value oracle returns $v_i(S) \in \mathbb{R}_+$ in unit time. 

We will use the prefix \emph{binary} before the names of function classes to denote that the valuation additionally has binary marginals, e.g., a function $v$ is binary $\XOS$ iff it is $\XOS$ and has binary marginals. The theorem below (proved in  Appendix \ref{appendix:prelims}) provides useful characterizations of binary $\XOS$ valuations. 
\begin{restatable}{theorem}{ThmBinaryXOSChar}
\label{theorem:xos-definitions}
A valuation $v : 2^{[m]} \mapsto \mathbb{R}_+$ is binary $\XOS$ iff it satisfies anyone of the following equivalent properties \\
($\mathrm{P}_1$): Function $v$ is $\XOS$ and it has binary marginals. \\
($\mathrm{P}_2$): Function $v$ has binary marginals and for every set $S \subseteq [m]$ there exists a subset $X \subseteq S$ with the property that $v(X) = |X| = v(S)$. \\
($\mathrm{P}_3$): Function $v$ can be expressed as a pointwise maximum of binary additive functions $\left\{ \ell_t: 2^{[m]} \mapsto \mathbb{R}_+  \right\}_{t=1}^L$, i.e., $v(S) = \max_{1 \leq t \leq L} \ \ell_t(S)$ for every $S \subseteq [m]$. Here, each function $\ell_t$ is additive and $\ell_t(g) \in \{0,1\}$ for every $g \in [m]$. \\
($\mathrm{P}_4$): There exists a family of subsets $\mathcal{F} \subseteq 2^{[m]}$ such that $v(S) = \max_{F \in \mathcal{F}} |S \cap F|$ for every set $S \subseteq [m]$.
\end{restatable}

\paragraph{Social welfare and Nash social welfare.} The social welfare $\SW(\cdot)$ of an allocation $\alloc = (A_1, \ldots ,A_n)$ is defined as the sum of the values that the agents derive from their bundles in $\alloc$, i.e., $\SW(\alloc) \coloneqq \sum_{i=1}^n v_i(A_i)$. 

The Nash social welfare $\NSW(\cdot)$ of an allocation $\alloc = (A_1,\ldots ,A_n)$ is defined as the geometric mean of the agents' values in $\alloc$, i.e., $\NSW(\alloc) \coloneqq \left( \prod_{i=1}^n v_i(A_i) \right)^{\frac{1}{n}}$. An allocation $\mathcal{N}=(N_1, \ldots, N_n)$ with the maximum possible Nash social welfare (among the set of all allocations) is referred to as a Nash optimal allocation.

Under binary $\XOS$ valuations, one can assume, without loss of generality, that welfare-maximizing allocations solely consist of non-wasteful bundles; the proof of this lemma is deferred to Appendix \ref{appendix:prelims}. 
\begin{restatable}{lemma}{PropNonWastefulAlloc}
\label{proposition:non-wasteful-allocation}
For any allocation $\mathcal{P}=(P_1,\ldots, P_n)$ among agents with binary $\XOS$ valuations, there exists an allocation $\mathcal{P}'= (P'_1,\ldots, P'_n)$ of non-wasteful bundles that has the same valuation profile as $\mathcal{P}$, i.e., $v_i(P'_i) = |P'_i| =v_i(P_i)$ for all agents $i \in [n]$. 
\end{restatable}

\section{Approximation Algorithm for Nash Social Welfare}
\label{section:nsw}
Our algorithm (Algorithm \ref{algo:apx-mnw}) computes an allocation $\alloc = (A_1, \ldots, A_n)$ in which, for each agent $i$, the envy is multiplicatively bounded towards the entire set of goods, $G_i$, allocated to agents with bundle size at least four times that of $i$. Specifically, with respect to allocation $\alloc$, write $H_i$ to denote the set of agents who have received a bundle of size at least four times that of $i$, and let $G_i = \big( \cup_{j \in H_i}A_j \big) \cup A_0 \cup A_i$; recall that $A_0$ denotes the set of unassigned goods in $\alloc$. We show that, under binary $\XOS$ valuations, if $v_i(A_i) > \frac{1}{2} v_i(G_i)$ for all agents $i$, then $\alloc$ achieves a constant-factor approximation guarantee for Nash social welfare. 

Algorithm \ref{algo:apx-mnw} finds such an allocation by iteratively updating the agents' bundles. In particular, if for an agent $i$ the envy requirement is not met (i.e., we have $v_i(A_i) \leq \frac{1}{2} v_i(G_i)$), then the algorithm finds a non-wasteful subset $X \subset G_i$ with twice the value of $A_i$, i.e., finds a subset $X \subset G_i$ with the property that $v_i(X) = |X| = 2 v_i(A_i)$. The algorithm then assigns $X$ to agent $i$, and updates the remaining bundles accordingly. Note that, under binary $\XOS$ valuations, such a subset $X$ can be computed efficiently (in Line \ref{step:update-A-i} of the algorithm): one can initialize $X = G_i$ and iteratively remove goods from $X$ until the desired property is achieved; recall $(\mathrm{P}_2$) in Theorem \ref{theorem:xos-definitions}. Also, with these updates, the algorithm maintains the invariant that the bundles assigned to the agents are non-wasteful. Indeed, the value of agent $i$ doubles after receiving subset $X$, and we show that the algorithm necessarily finds the desired allocation after at most a polynomial number of such value increments, i.e., the algorithm runs in polynomial time (Lemma \ref{lemma:time-complexity}). 

\floatname{algorithm}{Algorithm}
\begin{algorithm}[ht]
\caption{\textsc{Alg}} \label{algo:apx-mnw}
\textbf{Input:} Fair division instance  $\langle [m],[n], \{v_i\}_i \rangle$ with value-oracle access to the binary $\XOS$ valuations $v_i$s  \\
\textbf{Output:} Allocation $\mathcal{A} = (A_1, \ldots, A_n)$

  \begin{algorithmic}[1]
  		\STATE Compute an allocation $\alloc \coloneqq (A_1,\ldots, A_n)$ with $v_i(A_i) = |A_i| = 1$, for every agent $i \in [n]$. \COMMENT{Such an allocation $\alloc$ can be computed by finding a perfect matching between the agents $i$ and the goods valued by $i$.} \label{step:initial-matching}
  		\STATE Initialize $A_0 = [m]\setminus (\cup_{j=1}^n A_j)$ 
  			\STATE For each agent $i \in [n]$, initialize subset of agents $H_i \coloneqq \{j \in [n] \ : \ |A_j| > 4|A_i|\}$ and subset of goods $G_i \coloneqq \big( \cup_{j \in H_i}A_j \big) \cup A_0 \cup A_i$ \label{step:definition-G-H}
		\WHILE{there exists agent $i \in [n]$ such that $v_i(A_i) \leq \frac{1}{2} v_i(G_i)$} \label{step:while-loop}
			\STATE Find subset $X \subseteq G_i$ with the property that $v_i(X) = |X| = 2v_i(A_i)$  \COMMENT{Such a non-wasteful subset $X$ can be computed efficiently for binary $\XOS$ valuations} \label{step:update-A-i}
			\STATE Set $A_i = X$, and update $A_j \leftarrow A_j \setminus X$ for each $j \in H_i$ \label{line:update-middle} \label{line:update-begin}
			\STATE Set $A_0 = [m] \setminus ( \cup_{j=1}^n A_j)$ \label{line:update-end}
			\STATE Set $H_k = \{j \in [n] \ : \ |A_j| > 4|A_k|\}$ and $G_k = \big( \cup_{j \in H_k}A_j \big) \cup A_0 \cup A_k$, for each agent $k \in [n]$
		\ENDWHILE
		\STATE \textbf{return } $\alloc  = \left( A_1,\ldots, A_n \right)$ \label{line:return}
		\end{algorithmic}
\end{algorithm}

Write $\mathcal{N} = (N_1, \ldots, N_n)$ to denote a Nash optimal allocation for the given fair division instance. We will throughout assume that the optimal Nash welfare is positive, $\NSW(\mathcal{N}) >0$. In the complementary case, wherein $\NSW(\mathcal{N}) = 0$, returning an arbitrary allocation suffices.\footnote{Here, in fact, one can also maximize the Nash social welfare subject to the constraint that the maximum possible number of agents receive a good: write $n'$ to denote the size of the maximum-cardinality matching between the agents $i$ and the goods valued by $i$, and introduce $(n-n')$ ``dummy'' goods, any nonempty subset of which gives unit value to any agent. Approximating Nash social welfare in this modified instance (with binary $\XOS$ valuations) addresses the constrained version of the problem.} Note that the assumption $\NSW(\mathcal{N})>0$ and the fact that the valuations have binary marginals ensure that, for each agent $i$, the bundle $N_i$ contains a unit valued (by $i$) good. Hence, in Line \ref{step:initial-matching} of the algorithm we are guaranteed to find a matching wherein each agent is assigned a good of value one. 

The following lemma establishes an interesting property of the allocation $\alloc=(A_1,\ldots, A_n)$ returned by our algorithm. In particular, the lemma shows that---for any integer $\alpha \in \mathbb{Z}_+$---at most $n/\alpha$ agents $i$ receive a bundle $A_i$ of value less than $\frac{1}{18 \alpha}$ times $v_i(N_i)$. That is, in allocation $\alloc$, for any $\alpha \in \mathbb{Z}_+$ the number of  $18 \alpha$-suboptimal agents is at most $n/\alpha$. We will establish the approximation ratio of Algorithm \ref{algo:apx-mnw} (in Theorem~\ref{theorem:approximation-ratio} below) by invoking the lemma with dyadic values of $\alpha$.  


\begin{lemma}
\label{lemma:geometric-relationship}
Let $\alloc = (A_1,\ldots, A_n)$ be the allocation returned by Algorithm \ref{algo:apx-mnw}  (\textsc{Alg}) and $\mathcal{N} = (N_1,\ldots, N_n)$ be a Nash optimal allocation, with $\NSW(\mathcal{N})>0$. Also, for any integer $\alpha \in \mathbb{Z}_+$, let $X_\alpha \coloneqq \left\{ i \in [n] :  v_i(A_i) < \frac{1}{\const \alpha} v_i(N_i) \right\}$. Then, 
\begin{align*}
\left| X_\alpha \right| \leq \frac{n}{\alpha} 
\end{align*}
\end{lemma}
\begin{proof}
Throughout its execution \textsc{Alg} assigns a non-wasteful bundle to every agent (see Lines \ref{step:update-A-i} and \ref{line:update-begin}) and, hence, for the returned allocation $\alloc=(A_1, \ldots, A_n)$ we have $v_i(A_i) = |A_i|$, for all agents $i$. Also, we assume, without loss of generality, that the bundles in the Nash optimal allocation $\mathcal{N} = (N_1,\ldots, N_n)$ are non-wasteful; see Lemma \ref{proposition:non-wasteful-allocation}. 

Fix an integer $\alpha \in \mathbb{Z}_+$ and consider any agent $i \in X_\alpha$. Recall that $G_i$ contains the set of goods that (under allocation $\alloc$) are assigned among agents in $H_i \coloneqq \{j \in [n] \ : \ |A_j| > 4|A_i|\}$ . We begin by upper bounding the size of the intersection between $N_i$ and $G_i$;  in particular, this bound shows that, in $\alloc$, not too many goods from the optimal bundle $N_i$ can get assigned among agents in $H_i$. Towards this, note that the termination condition of the while-loop (Line ~\ref{step:while-loop}) ensures that the returned non-wasteful bundle $A_i$ satisfies $v_i(G_i) < 2v_i(A_i)=2|A_i|$. Therefore, using Proposition \ref{proposition:non-wasteful-subsets} and the fact that $N_i$ is non-wasteful we get 
\begin{align}
\label{ineq-1}
|N_i \cap G_i| & = v_i(N_i \cap G_i) \leq v_i(G_i) < 2|A_i| 
\end{align}

Also, since agent $i \in X_\alpha$, the cardinality of $N_i$ is more than $\const \alpha$ times that of $A_i$: $|N_i| = v_i(N_i) > \const \alpha \ v_i(A_i) = \const \alpha \  |A_i|$. This observation and inequality (\ref{ineq-1}) imply that $N_i$ has a sufficiently large intersection with $G_i^c \coloneqq [m] \setminus G_i$
\begin{align}
\label{ineq-3}
|N_i \cap G_i^c| = |N_i| - |N_i \cap G_i| > \const \alpha  |A_i| - 2|A_i| \geq \constMinTwo \alpha |A_i|
\end{align}
Indeed, $G_i^c$ is the set of goods that, in allocation $\alloc$, are assigned among the agents $j \in [n] \setminus (H_i \cup \{ i \})$, i.e., among the agents $j \neq i$ with bundles of value $v_j(A_j)  =  |A_j| \leq 4 |A_i|$. 

To establish the desired upper bound on the size of $X_\alpha$, we partition it into subsets. Specifically, for each $0 \leq k \leq \lfloor \log_4 m \rfloor$, define set 
\begin{align*}
X^{k}_\alpha \coloneqq \{i \in X_\alpha \ : \ 4^k \leq v_i(A_i) < 4^{k+1} \}.
\end{align*} 
That is, $X^{k}_\alpha$ is the set of agents for whom the ratio between assigned value and the optimal value is less than $\frac{1}{18\alpha}$ (i.e., $i \in X_\alpha$) and the assigned value is in the range $[4^k, 4^{k+1})$. We note that with $k$ between $0$ and $\lfloor \log_4 m \rfloor$, the subsets $X^{k}_\alpha$s, partition $X_\alpha$. In particular, initially in $\Alg$ (see Line \ref{step:initial-matching}) each agent achieves a value of one; recall the assumption that $\NSW(\mathcal{N}) >0$ and, hence, there exists a matching wherein each agent is assigned a nonzero valued good. Furthermore, during the execution of \textsc{Alg} the agents' valuations inductively continue to be at least one: consider any iteration of the while-loop and let $\widehat{i}$ be the agent that receives a new bundle $X$ in the iteration  (see Lines \ref{step:update-A-i} and \ref{line:update-begin}). The selection criterion of $X$ ensures that the valuation of $\widehat{i}$ in fact doubles. For any agent $ j \in H_{\widehat{i}}$, before the update in Line \ref{line:update-middle} we have $v_j(A_j) = |A_j| > 4 |A_{\widehat{i}}| = 2 |X|$ and, hence, even after the update ($A_j \leftarrow A_j \setminus X$) agent $j$'s value continues to be at least one. Finally, for each remaining agent (in the set $[n] \setminus \left( H_{\widehat{i}} \cup \{ \ \widehat{i} \ \} \right)$) its bundle remains unchanged. Hence, for the returned allocation $\alloc=(A_1, \ldots, A_n)$ we have $v_i(A_i) = |A_i| \geq 1$. Also, the fact that the marginals of the valuation $v_i$ are binary implies $v_i([m]) \leq m$, i.e., $v_i(A_i) \leq m$. Therefore, the bounds $1 \leq v_i(A_i) \leq m$ (for all agents $i$) imply that the subsets $X^{k}_\alpha$s, with $0 \leq k \leq \lfloor \log_4 m \rfloor$, partition $X_\alpha$; in particular, $\sum_{k=0}^{\lfloor \log_4 m \rfloor} |X^k_\alpha| = |X_\alpha|$. 

Furthermore, for each agent $i \in X^{k}_\alpha$ we have 
\begin{align*}
|N_i \cap G_i^c| & > \constMinTwo \alpha \ |A_i| \tag{via inequality (\ref{ineq-3})} \\
& \geq \constMinTwo \alpha  \  4^k \tag{since $i \in X^k_\alpha$} \\
& = \alpha \ 4^{k+2}.
\end{align*}
Therefore, for each $k$, the set of goods $D_k \coloneqq \bigcup_{i \in X^k_\alpha} (N_i \cap G_i^c)$ satisfies $|D_k| \geq \alpha 4^{k+2} \  |X^k_\alpha|$. That is, for each $k$, and whenever $X^k_\alpha \neq \emptyset$, the size of set $D_k$ is at least a positive integer multiple of $4^{k+2}$. Also, note that for each good $\overline{g} \in D_k$, we have (by definition of $D_k$) that $\overline{g} \in N_{\overline{i}} \cap G^c_{\overline{i}}$ for some $\overline{i} \in X^k_\alpha$.  These containments (and the definition of $G^c_{\overline{i}}$) ensure that $\overline{g} \in A_j$,\footnote{Recall that $A_0 \subseteq G_{\overline{i}}$ and, hence, $A_0 \cap G^c_{\overline{i}} = \emptyset$.} for some agent $j \in [n]$ with the property that $|A_j| \leq 4 |A_{\overline{i}}| < 4^{k+2}$. That is, for each $k$ (with $X^k_\alpha \neq \emptyset$), the cardinality of $D_k$ is a positive integer multiple of $4^{k+2}$ and (under $\alloc$) the goods in $D_k$ must be assigned to agents with bundles of size at most $4^{k+2}$. These two properties ensure that (in allocation $\alloc$) a sufficiently large number of bundles are necessarily required to cover the set of goods $\bigcup_k D_k = \bigcup_k \bigcup_{i \in X^k_\alpha} (N_i \cap G_i^c) = \bigcup_{i \in X_\alpha} (N_i \cap G_i^c)$. Specifically, write $t \in \mathbb{Z}_+$ to denote  the number of agents that have been assigned (under allocation $\alloc$) at least one good from $\bigcup_k D_k$ (i.e., $t\coloneqq \left| \left\{ j \in [n] : A_j \cap (\cup_k D_k)  \neq \emptyset \right\} \right|$), then Proposition \ref{prop:constrained-packing} (stated and proved in Appendix \ref{appendix:nsw}) gives us 

\begin{align*}
t \geq \sum_{k=0}^{\lfloor \log_4 m \rfloor} \frac{|D_k|}{4^{k+2}} \geq \sum_{k=0}^{\lfloor \log_4 m \rfloor} \frac{\alpha 4^{k+2} \cdot |X^k_\alpha|}{4^{k+2}} = \alpha \ \sum_{k=0}^{\lfloor \log_4 m \rfloor} |X^k_\alpha| = \alpha |X_\alpha|.
\end{align*}
However, the number of agents $t$ cannot be more than $n$. Hence, the stated claim follows $n \geq t \geq \alpha |X_\alpha|$.
\end{proof}

The allocation $\alloc = (A_1, A_2, \ldots, A_n)$ returned by Algorithm \ref{algo:apx-mnw} can be made complete by allocating the (unassigned) goods in $[m] \setminus \cup_{i=1}^n A_i$ arbitrarily. Doing this would not affect the approximation guarantee. \\

The following lemma establishes the time complexity of $\Alg$. 

\begin{restatable}{lemma}{timeComplexityLemma}
\label{lemma:time-complexity}
For any given fair division instance with $n$ agents, $m$ goods, and value-oracle access to the binary $\XOS$ valuations, Algorithm \ref{algo:apx-mnw} ($\Alg$) returns an allocation in time that is polynomial in $n$ and $m$. 
\end{restatable}
\begin{proof}
Initially, the Nash social welfare of allocation $\alloc=(A_1, \ldots, A_n)$ is one, since $v_i(A_i) = |A_i| = 1$, for all agents $i \in [n]$; see Line \ref{step:initial-matching} and recall the assumption that $\NSW(\mathcal{N})>0$. We will prove that in every while-loop iteration of $\Alg$, the Nash social welfare of $\alloc$ increases by at least a factor of $\left(1+\frac{1}{4m+1} \right)^{\frac{1}{n}}$. That is, after every $(4m+1)n$ iterations, the Nash social welfare of the maintained allocation $\alloc$ increases by at least a factor of two. Also, note that, under binary $\XOS$ valuations, the optimal Nash social welfare is at most $m/n$. Therefore, the algorithm necessarily terminates after $\mathcal{O}(mn \ \log (m/n))$ iterations. This establishes the polynomial-time complexity of $\Alg$. 

Hence, we complete the runtime analysis by showing that after every update the Nash social welfare increases by at least a factor of $\left( 1+\frac{1}{4m+1} \right)^{\frac{1}{n}}$. Towards this, fix a while-loop iteration, and let $\alloc = (A_1, A_2, \ldots, A_n)$ and $\alloc' = (A'_1, A'_2, \ldots, A'_n)$, respectively, denote the non-wasteful allocations before and after the update steps (Lines \ref{line:update-begin} and \ref{line:update-end}). Also, let $x$ be that agent that receives a new bundle $X$ in Line \ref{line:update-begin}. Note that $|A'_{x}| = |X| = 2 |A_{x}|$ and $A'_j = A_j \setminus X$ for all $j \in H_{x}$. 

To prove that $\frac{\NSW(\alloc')}{\NSW(\alloc)} \geq \left( 1+\frac{1}{4m+1} \right)^{\frac{1}{n}}$, we start with the valuation profile of $\alloc'$, obtain a vector $(\eta_1,\ldots ,\eta_n) \in \mathbb{Z}_+^{n}$  such that $\NSW(\alloc') \geq \left( \prod_{i=1}^n \eta_i \right)^{1/n}$, and show that even $\left( \prod_{i=1}^n \eta_i \right)^{1/n}$ is more than the desired factor times $\NSW(\alloc)$. 

In particular, write $y \coloneqq \argmin_{j \in H_x} |A'_j|$ and \emph{initialize} $\eta_i = |A'_i|$ for all $i \in [n]$. Then, iteratively for each agent $j \in H_x \setminus \{ y\}$, we update $\eta_j \leftarrow \eta_j + |A_j \cap X|$ and $\eta_y \leftarrow \eta_y - |A_j \cap X|$.\footnote{If $A_j \cap X = \emptyset$, then $\eta_j$ remains unchanged.} Note that $y = \argmin_{j \in H_x} |A'_j|$ and, hence, each update increments a higher component ($\eta_j$) and decrements a smaller one ($\eta_y$), by the same amount. Therefore, each update reduces $\left( \prod_{i=1}^n \eta_i \right)^{1/n}$ and maintains the invariant $\NSW(\alloc') \geq \left( \prod_{i=1}^n \eta_i \right)^{1/n}$.

Also, recall that, for all $j \in H_x$, we have $A'_j = A_j \setminus X$ (i.e., $|A'_j| = |A_j| - |A_j \cap X|$). Hence, the updates give us $\eta_j = |A_j|$ for all $j \in H_x \setminus \{ y\}$. Furthermore, after the updates, agents $x$ and $y$ satisfy: $\eta_x = |A'_x| = 2|A_x|$ and $\eta_y = |A'_y| - \left( \sum_{j \in H_x \setminus \{y\}} |A_j \cap X| \right) = |A_y| - \left( \sum_{j \in H_x} |A_j \cap X| \right) \geq |A_y| - |X|$. For each remaining agent $z \in [n] \setminus (H_x \cup \{ x \})$, the bundle remains unchanged in the loop, $A_z = A'_z$, and, hence, $\eta_z = |A'_z| = |A_z|$. Therefore, at the end, we have $\eta_j = |A_j|$ for all $j \in [n] \setminus \{x, y\}$, $\eta_x = 2 |A_x|$, and $\eta_y \geq |A_y|  - |X| = |A_y| - 2|A_x|$. Using these bounds on $\eta_j$s we obtain 
\begin{align}
\frac{\left( \prod_{i=1}^n \eta_i \right)^{\frac{1}{n}}}{\NSW(\alloc)} = \left( \prod_{i=1}^n \frac{\eta_i}{|A_i|} \right)^{\frac{1}{n}} = \left(  \frac{\eta_x \cdot \eta_y}{ |A_x| |A_y|} \right)^{\frac{1}{n}}  \geq \left( \frac{2|A_x|  \cdot \left( |A_y| - 2 |A_x|\right)}{|A_x| |A_y|} \right)^{\frac{1}{n}} = \left(  2 \ \left( 1-\frac{2|A_x|}{|A_y|}  \right)\right)^{\frac{1}{n}} \label{ineq:long-tail}
\end{align}

Since $y \in H_x$, we have $|A_y| \geq 4 |A_x| +1$. Extending inequality (\ref{ineq:long-tail}) with this bound on $|A_y|$ we get 
\begin{align*}
\frac{\left( \prod_{i=1}^n \eta_i \right)^{\frac{1}{n}}}{\NSW(\alloc)} \geq  \left(  2 \ \left( 1-\frac{2|A_x|}{|A_y|}  \right)\right)^{\frac{1}{n}} \geq  \left(  2 \ \left( 1-\frac{2|A_x|}{4|A_x| +1}  \right)\right)^{\frac{1}{n}} = \left(  1 + \frac{1}{4|A_x| + 1} \right)^{\frac{1}{n}} \geq \left(  1 + \frac{1}{4m + 1} \right)^{\frac{1}{n}}.
\end{align*}
The last inequality follows from the fact that $|A_x| \leq m$. 

As mentioned previously, $\NSW(\alloc') \geq \left( \prod_{i=1}^n \eta_i \right)^{1/n}$ and, hence, $\frac{\NSW(\alloc')}{\NSW(\alloc)} \geq \left( 1+\frac{1}{4m+1} \right)^{\frac{1}{n}}$. Overall, we get that, in each iteration, the Nash social welfare increases by at least a factor of $\left(1+\frac{1}{4m+1} \right)^{\frac{1}{n}}$, and this completes the proof.  
\end{proof}

The following theorem is our main result for Nash social welfare.
\begin{theorem}
\label{theorem:approximation-ratio}
For binary $\XOS$ valuations and in the value-oracle model, there exists a polynomial-time $\approxConst$-approximation algorithm for the Nash social welfare maximization problem.  
\end{theorem}
\begin{proof}
Let $\alloc = (A_1,\ldots, A_n)$ be the (non-wasteful) allocation returned by $\Alg$, and $\mathcal{N} = (N_1,\ldots, N_n)$ be a (non-wasteful) Nash optimal allocation. As mentioned previously, under the condition that optimal Nash welfare $\NSW(\mathcal{N}) >0$, every agent necessarily receives a value of at least one in the allocation $\alloc$, i.e., $v_i(A_i) \geq 1$ for all $i$. In addition, $v_i(N_i) \leq v_i([m]) \leq m$; the last inequality follows from the fact the valuations have binary marginals. Hence, for all agents $i$, we have $v_i(A_i) \geq \frac{1}{m} v_i(N_i)$. 

Next, we partition the set of agents based on the ratio of their assigned value, $v_i(A_i)$, and their optimal value, $v_i(N_i)$. Specifically, define set $Y_{2^d} \coloneqq \left\{ i \in [n] \ : \frac{1}{2^{d+1}} \frac{v_i(N_i)}{18} \leq v_i(A_i) < \frac{1}{2^{d}} \frac{v_i(N_i)}{18} \right\}$, for each integer $d \in \{0,1,\ldots, \lceil \log  {m} \rceil \}$.  Since $v_i(A_i) \geq \frac{1}{m} v_i(N_i)$ for all $i$, the remaining agents $i' \in Y' \coloneqq [n] \setminus \left( \bigcup_{d=0}^{\lceil \log  {m} \rceil} \ Y_{2^d} \right)$ satisfy $v_{i'}(A_{i'}) \geq \frac{1}{18} v_{i'}(N_{i'})$. Indeed, the subsets $Y_{2^d}$s and $Y'$ partition the set of agents, and $|Y'| + \sum_{d=0}^{ \lceil \log  {m} \rceil } |Y_{2^d}| =n$. 

Note that, with $\alpha = 2^d$, we have $Y_\alpha \subseteq X_\alpha$; here, set $X_\alpha$ is defined as in Lemma \ref{lemma:geometric-relationship}. Hence, invoking this lemma we get 
\begin{equation}
\label{size-ineq}
|Y_{2^d}| \leq \frac{n}{2^d} \qquad \text{for all $0 \leq d \leq \lceil \log  {m} \rceil$}
\end{equation}

Write $\pi(S) \coloneqq \prod_{i \in S} \frac{v_i(A_i)}{v_i(N_i)}$, if the subset of agents $S \neq \emptyset$, and $1$ otherwise. Now, towards establishing the approximation ratio, consider
\begin{align*}
\frac{\NSW(\alloc)}{\NSW(\mathcal{N})} & = \left( \prod_{i=1}^n \frac{v_i(A_i)}{v_i(N_i)} \right)^{1/n} = \left( \pi(Y') \prod_{d=0}^{\lceil \log  {m} \rceil} \pi \left(Y_{2^d} \right) \right)^{1/n} \\
& \geq \left( \left( \frac{1}{18} \right)^{|Y'|} \prod_d \pi \left(Y_{2^d} \right) \right)^{1/n} \tag{since $v_i(A_i) \geq \frac{1}{18} v_i(N_i)$ for all $i \in Y'$} \\
& \geq \left( \left( \frac{1}{18} \right)^{|Y'|} \prod_d  \left( \frac{1}{18 \ 2^{d+1}} \right)^{|Y_{2^d}|}  \right)^{1/n} \tag{since $v_i(A_i) \geq \frac{1}{18 \ 2^{d+1} } v_i(N_i)$ for all $i \in Y_{2^d}$} \\
& = \frac{1}{18} \left( \prod_d  \left( \frac{1}{2^{d+1}} \right)^{|Y_{2^d}|}  \right)^{1/n} \tag{since $|Y'| + \sum_{k} |Y_{2^d}| =n$} \\
& = \frac{1}{18} \left( \prod_d  \left( \frac{1}{2^{d+1}} \right)^{|Y_{2^d}|/n}  \right) \\
& \geq \frac{1}{18} \left( \prod_d  \left( \frac{1}{2^{d+1}} \right)^{\frac{1}{2^d}}  \right) \tag{via inequality (\ref{size-ineq})}
\end{align*}

Proposition \ref{prop:power-product} (proved in Appendix \ref{appendix:nsw}) shows that the product $\prod_d  \left( \frac{1}{2^{d+1}} \right)^{\frac{1}{2^d}} \geq \frac{1}{16}$. Therefore, the stated approximation bound follows 
\begin{align*}
\frac{\NSW(\alloc)}{\NSW(\mathcal{N})}
\geq \frac{1}{\const} \prod_{d=0}^{\lceil \log  {m} \rceil}  \left( \frac{1}{2^{d+1}} \right)^{\frac{1}{2^d}}
\geq \frac{1}{\const} \cdot \frac{1}{16} = \frac{1}{288}.
\end{align*}
\end{proof}

\section{Social Welfare and Groupwise Maximin Share Guarantees}
This section shows that, for binary $\XOS$ valuations, Algorithm \ref{algo:apx-mnw} ($\Alg$) achieves constant-factor approximations for social welfare (Theorem \ref{theorem:sw}) and the groupwise maximin share guarantee (Theorem \ref{theorem:gmms}).

\begin{theorem}
\label{theorem:sw}
For any fair division instance with binary $\XOS$ valuations, Algorithm \ref{algo:apx-mnw} ($\Alg$) computes (in the value-oracle model) an allocation with social welfare at least $1/(3+2 \sqrt{2})$ times the optimal (social welfare). 
\end{theorem}
\begin{proof}
Let $\alloc^* = (A^*_1,\ldots, A^*_n)$ be a social welfare maximizing allocation. We can assume, without loss of generality, that $\alloc^*$ is composed of non-wasteful bundles (Lemma \ref{proposition:non-wasteful-allocation}). Also, let $\alloc = (A_1,\ldots, A_n)$ be the non-wasteful allocation returned by $\Alg$.

We partition the set of agents into two sets: $\mathcal{H} \coloneqq \{ i \in [n] : v_i(A_i) \geq \frac{1}{2+\sqrt{2}} v_i(A^*_i) \}$ and $\mathcal{L} \coloneqq \{ i \in [n] : v_i(A_i) < \frac{1}{2+\sqrt{2}} v_i(A^*_i) \}$. Note that $\mathcal{H} \cup \mathcal{L} = [n]$. We establish the stated approximation guarantee by analyzing the social welfare of agents in $\mathcal{H}$ and $\mathcal{L}$ separately. 

For each agent $i \in \mathcal{H}$, by definition, we have $v_i(A^*_i) \leq (2+\sqrt{2}) \ v_i(A_i)$. Therefore, 
\begin{align}
\sum_{i \in \mathcal{H}} v_i(A^*_i) \leq (2+\sqrt{2}) \ \sum_{i \in \mathcal{H}} v_i(A_i) \leq (2+\sqrt{2}) \ \sum_{i \in [n]} v_i(A_i) \label{ineq:welfare-of-h}
\end{align}
We will now establish a similar inequality for the optimal social welfare of agents in $\mathcal{L}$. The while-loop condition of $\Alg$ ensures that, at termination, $v_i(G_i) < 2 v_i(A_i)$, for each agent $i \in [n]$. Hence, for an agent $i \in \mathcal{L}$, we have $v_i(G_i) < 2 v_i(A_i) < \frac{2}{2+\sqrt{2}} v_i(A^*_i)$. Furthermore, the monotonicity of $v_i$ gives us $v_i(G_i \cap A^*_i) < \frac{2}{2+\sqrt{2}} v_i(A^*_i) $ for each $i \in \mathcal{L}$.  

Write $G_i^c \coloneqq [m] \setminus G_i$ and note that the last inequality implies $v_i(G_i^c \cap A^*_i) > \frac{\sqrt{2}}{2+\sqrt{2}} v_i(A^*_i)$, for each $i \in \mathcal{L}$; otherwise, we would have $v_i(G_i^c \cap A^*_i) + v_i(G_i \cap A^*_i) < v_i(A^*_i)$, which contradicts the fact that $v_i$ is subadditive (in fact, $\XOS$). Hence, for each $i \in \mathcal{L}$, 
\begin{align}
v_i(A^*_i) < (1+\sqrt{2}) \ v_i(G_i^c \cap A^*_i) = (1+\sqrt{2}) \ |G_i^c \cap A^*_i| \label{ineq:interim-sw}
\end{align}
For the last equality, recall that $A^*_i$ is non-wasteful and, hence, its subsets are non-wasteful as well (Proposition \ref{proposition:non-wasteful-subsets}). Summing inequality (\ref{ineq:interim-sw}) over the agents $i \in \mathcal{L}$, we get 
\begin{align}
\label{ineq:penultimate-welfare-of-l}
\sum_{i \in \mathcal{L}} v_i(A^*_i) < (1+\sqrt{2}) \ \sum_{i \in \mathcal{L}} |G_i^c \cap A^*_i|
\end{align}
To simplify equation (\ref{ineq:penultimate-welfare-of-l}), we use the following two observations. First, the sets $(G_i^c \cap A^*_i)$s are pairwise disjoint across the agents. Second, each good $\widetilde{g} \in (G_i^c \cap A^*_i) \subseteq G_i^c =[m] \setminus G_i$ is allocated to some agent $j \in [n]$ under the allocation $\alloc$. This follows from the fact that $\widetilde{g} \notin A_0 \subseteq G_i$. Therefore, the set $\cup_{i \in \mathcal{L}} (G_i^c \cap A^*_i)$ must be contained in $\cup_{i \in [n]} A_i$, i.e., $\sum_{i \in \mathcal{L}}  |G_i^c \cap A^*_i| \leq \sum_{i \in [n]} |A_i|$. This inequality and equation (\ref{ineq:penultimate-welfare-of-l}) give us 
\begin{align}
\label{ineq:welfare-of-l}
\sum_{i \in \mathcal{L}} v_i(A^*_i) < (1+\sqrt{2}) \ \sum_{i \in \mathcal{L}}  |G_i^c \cap A^*_i| \leq (1+\sqrt{2}) \ \sum_{i \in [n]} |A_i| = (1+\sqrt{2}) \ \sum_{i \in [n]} v_i(A_i)
\end{align}
Finally, summing inequalities (\ref{ineq:welfare-of-h}) and (\ref{ineq:welfare-of-l}), we obtain the stated approximation bound
\begin{align*}
(3+2\sqrt{2})  \ \sum_{i \in [n]} v_i(A_i) & > \sum_{i \in \mathcal{H} \cup \mathcal{L}} v_i(A^*_i) = \sum_{i \in [n]} v_i(A^*_i).
\end{align*}
\end{proof}

\subsection{Groupwise Maximin Shares under Binary XOS valuations}
\label{section:gmms}

This section shows that the allocations computed by $\Alg$ achieve an approximate maximin share guarantee among all subgroup of agents. To formally specify this fairness guarantee we first define restricted maximin share for fair division instances $\langle [m], [n], \{v_i \}_{i=1}^n \rangle$: given parameter $r \in \mathbb{Z}_+$, subset of goods $S \subseteq [m]$, and agent $i \in [n]$, write $\mu_i^r(S) \coloneqq \max_{(P_1, \ldots, P_r)} \ \min_{1 \leq j \leq r} \ v_i(P_j)$; here, the $\max$ is considered over all the $r$-partitions of $S$. 

Note that $\mu_i^n([m])$ corresponds to the maximin share ($\MMS$) of agent $i$, and an allocation $(A_1, \ldots, A_n)$ is said to be $\alpha$-$\MMS$ iff $v_i (A_i) \geq \alpha \ \mu_i^n([m])$ for all agents $i \in [n]$. Strengthening this fairness notion, the groupwise maximin share guarantee requires that the (restricted) maximin share is achieved among all subsets of agents. Specifically, for an allocation $\alloc=(A_1, \ldots, A_n)$ and each agent $i \in [n]$, write 
\begin{align*}
\GMMS_i (\alloc) \coloneqq \max_{R \subseteq [n]: R \ni i} \ \mu_i^{|R|} \left(\bigcup_{j \in R} A_j \  \cup A_0\right).
\end{align*}
Here, $A_0$ is the set of unassigned goods, $ A_0 = [m] \setminus \left( \cup_{j \in [n]} A_j \right)$. Observe that, in contrast to the maximin share, $\mu_i^n([m])$, the threshold $\GMMS_i(\cdot)$ depends on the allocation at hand. 
\begin{definition}[Approximate $\GMMS$ allocation]
An allocation $\alloc=(A_1, \ldots, A_n)$ is said to be $\alpha$-$\GMMS$ iff $v_i(A_i) \geq \alpha\  \GMMS_i ( \alloc)$ for all agents $i \in [n]$.  
\end{definition}

Indeed, if an allocation $\alloc$ is $\alpha$-$\GMMS$, then for every subset of agents $R \subseteq [n]$, each agent $i \in R$ approximately achieves the maximin share obtained by solely considering agents in $R$ and the goods assigned to them, along with $A_0$. Here, including the unallocated goods $A_0$ in the definition of $\GMMS_i(\alloc)$ ensures that this threshold can be meaningfully applied in the context of partial allocations as well (i.e., with $A_0 \neq \emptyset$).  In fact, the definition ensures that given any partial, $\alpha$-$\GMMS$ allocation, one can arbitrarily assign $A_0$ among the agents to obtain a complete allocation that continues to be $\alpha$-$\GMMS$.  

The next result establishes that the allocation computed by Algorithm \ref{algo:apx-mnw} ($\Alg$) achieves the $\GMMS$ guarantee with $\alpha = 1/6$. 

\begin{theorem}
\label{theorem:gmms}
Given any given fair division instance with binary XOS valuations, $\Alg$ returns an allocation that is $\frac{1}{6}$-$\GMMS$.
\end{theorem}
\begin{proof}
Assume, towards a contradiction, that the non-wasteful allocation returned by the algorithm, $\alloc = (A_1,\ldots, A_n)$ is not $\frac{1}{6}$-$\GMMS$. That is, there exists a subset of agents $R \subseteq [n]$ and an agent ${i} \in R$ for which, $v_i(A_i) <  \frac{1}{6} \ \mu_i^{|R|} \left(\bigcup_{j \in R} A_j \  \cup A_0\right)$. Write $r \coloneqq |R|$. 

From the definition of the restricted maximin share $\mu_i^r \left( \bigcup_{j \in R} A_j \  \cup A_0\right)$, we know what there exists an $r$-partition $(B_1, B_2, \ldots, B_r)$ of the set of goods $\left( \bigcup_{j \in R} A_j \  \cup A_0\right)$ such that, for each $k \in [r]$,  we have $v_i(B_k) = |B_k| \geq \mu_i^r \left(\bigcup_{j \in R} A_j \  \cup A_0\right)$. Hence, $v_i(A_i) < \frac{1}{6} |B_k|$ for each $1 \leq k \leq r$. Since the allocation $\alloc$ is non-wasteful, the previous inequality reduces to 
\begin{align}
\label{ineq-gmms-1}
6 \ |A_i| < |B_k| \qquad \text{ for all $1 \leq k \leq r$} 
\end{align}

For the agent $i$ under consideration, we have $G_i = \big( \cup_{j \in H_i}A_j \big) \cup A_0 \cup A_i$ (see $\Alg$) and $v_i(G_i) \leq 2v_i(A_i) = 2 |A_i|$. 
Furthermore, the fact that, for each $1\leq k \leq r$, the subset $B_k$ is non-wasteful (with respect to $v_i$) gives us 
\begin{align}
|B_k \cap G_i| & = v_i(B_k \cap G_i)  \tag{via Proposition \ref{proposition:non-wasteful-subsets}} \nonumber \\
& \leq v_i(G_i) \tag{since $v_i$ is monotonic} \nonumber \\
& \leq 2 |A_i| \label{ineq:half-inter}
\end{align}
Write $G^c_i \coloneqq [m] \setminus G_i$ and note that (in allocation $\alloc$) each good in $G^c_i$ is assigned to some agent $j \in H^c_i \coloneqq [n] \setminus H_i$; this follows from the fact that $G_i = \big( \cup_{j \in H_i}A_j \big) \cup A_0 \cup A_i$. Furthermore, note that the subsets $B_k$s partition $\bigcup_{j \in R} A_j \  \cup A_0$. These two observations ensure that each good $\widetilde{g} \in G^c_i \cap B_k$ is contained in some bundle $A_j$ with $j \in H^c_i \cap R$. Hence, the pairwise disjoints subsets $\left\{ G^c_i \cap B_k \right\}_{k=1}^r$ are contained in $\cup_{j \in \left( H^c_i \cap R \right)} A_j$, and we have 
\begin{align}
\sum_{k=1}^r |G^c_i \cap B_k| &  \leq \sum_{j \in H^c_i \cap R} |A_j| \nonumber \\
& \leq \sum_{j \in H^c_i \cap R} 4 |A_i| \tag{since $|A_j| \leq 4 |A_i|$ for each $j \in H^c_i$} \nonumber \\
& \leq 4 r |A_i| \label{ineq:gmms-contra}
\end{align}

On the other hand, for each $1 \leq k \leq r$, we have the following lower bound 
\begin{align*}
|G^c_i \cap B_k| & = |B_k| - |B_k \cap G_i| \\ 
& \geq |B_k| - 2 |A_i| \tag{via inequality (\ref{ineq:half-inter})} \\ 
& > 4 |A_i| \tag{via inequality (\ref{ineq-gmms-1})} 
\end{align*}
Summing the last inequality, over $k \in [r]$, gives us $\sum_{k=1}^r |G^c_i \cap B_k|  > 4r |A_i|$. This contradicts equation (\ref{ineq:gmms-contra}) and, hence, shows that for every subset $R$ and agent $i \in R$ the restricted maximin share guarantee is achieved at least with a factor of $1/6$. 
\end{proof}

\section{Hardness of Approximation for Binary $\XOS$ Valuations}
This section establishes the $\mathrm{APX}$-hardness of maximizing Nash social welfare in fair division instances with binary $\XOS$ valuations. 
This inapproximability holds even if the agents' (binary $\XOS$) valuations are identical and admit a succinct representation. We obtain the hardness result by developing an approximation preserving reduction from the following gap version of the independent set problem in $3$-regular graphs. 

\begin{theorem}[\cite{chlebik2003hardness}]
\label{theorem:max-3-is-gap-hardness}
Given a $3$-regular graph $\mathcal{G}$ and a threshold $\tau$, it is $\mathrm{NP}$-hard to distinguish between 
\begin{itemize}
\item {\rm YES} Instances: The size of the maximum independent set in $\mathcal{G}$ is at least $\tau$. 
\item {\rm NO} Instances: The size of the maximum independent set in $\mathcal{G}$ is at most $\frac{94}{95} \tau$.
\end{itemize} 
\end{theorem}

Note that in the theorem above, for the given graph, either the maximum independent set is of size at least $\tau$ or it is at most $\frac{94}{95} \tau$, i.e., here we have a promise problem. Our hardness result is established next.

\begin{theorem}
\label{theorem:apx-hard}
For fair division instances with (identical) binary $\XOS$ valuations, it is $\mathrm{NP}$-hard to approximate the maximum Nash social welfare within a factor of $1.0042$.
\end{theorem}
\begin{proof}
We will show that, given a $c$-approximation algorithm for maximizing Nash social welfare, with an appropriate constant $c>1$, one can distinguish between the {\rm YES} and {\rm NO} instances of Theorem \ref{theorem:max-3-is-gap-hardness}.  

The reduction is as follows. For the given $3$-regular graph $\mathcal{G}=(V, E)$ and parameter $\tau$, we construct a fair division instance with $\tau$ agents and $m = |E|$ indivisible goods, one good corresponding to each edge of $\mathcal{G}$. For ease of presentation, the set of goods is also denoted by $E$. All agents in this fair division instance have an identical valuation $f : 2^{E} \mapsto \mathbb{R}_+$. Write $\delta(v) \coloneqq \{ e= (u, v) \in E : u \in V \}$ denote the set of edges incident on a vertex $v \in V$. With these size-$3$ subsets of edges in hand, we define, for every set $S \subseteq E$, the valuation as follows $f(S) \coloneqq \max_{v \in V} |S \cap \delta(v)|$. Valuation $f$ is indeed a binary $\XOS$ function; see property ($\mathrm{P}_4$) in Theorem \ref{theorem:xos-definitions}. 

In the given graph $\mathcal{G}$, the maximum independent set is of size at least  $\tau$ or at most $\frac{94}{95} \tau$. We will show that, analogously,   the optimal Nash social welfare (in the constructed fair division instance) will either be high or low; the multiplicative gap here will rule out a $c$-approximation algorithm.

First, we consider the case wherein the maximum independent set $\mathcal{I}$ is of size at least $\tau$. In such a setting, each agent in the fair division instance can be allocated the set of edges $\delta(u)$ for a vertex $u \in \mathcal{I}$. Since the vertices in $\mathcal{I}$ do not share an edge and $|\mathcal{I}|$ is at least the number of agents $\tau$, every agent gets a (disjoint subset) bundle of size three; recall that the graph is $3$-regular. Therefore, the optimal Nash social welfare in this case is $3$.\footnote{Note that, by definition, the value of $f$ is upper bounded by $3$, i.e., the Nash social welfare cannot be greater than $3$.} 

In the second case the maximum independent set $\mathcal{I}$ is of size is at most $\frac{94}{95} \tau$. Consequently, under any allocation of edges (goods), at most $\frac{94}{95} \tau$ agents receive a bundle of value $3$; this follows from the observation that, given any allocation in which $t \in \mathbb{Z}_+$ agents have a bundle of value $3$, we can construct an independent set of size $t$. Next, note that all the agents who do not receive a bundle of value $3$, can have value at most $2$ for their bundle; there are at least $\tau - \frac{94 \tau}{95} = \frac{\tau}{95}$ such agents. Therefore, in this case, the optimal Nash social welfare is upper bounded by 
\begin{align*}
\left(3^{ \frac{94 \tau}{95} } \ 2^{ \frac{\tau}{95} } \right)^{\frac{1}{\tau}} = 3^{\frac{94}{95}} \ 2^{\frac{1}{95}}.
\end{align*}

Hence, the two cases can be distinguished from each other via a $c$-approximation algorithm for Nash social welfare, with $c \leq \frac{3}{3^{\frac{94}{95}} \ 2^{\frac{1}{95}} } = \left( \frac{3}{2} \right)^{\frac{1}{95}}$. Since $\left( \frac{3}{2} \right)^{\frac{1}{95}} > 1.0042$, we obtain the stated result that is $\mathrm{NP}$-hard to approximate the maximum Nash social welfare within a factor of $1.0042$ under binary $\XOS$ valuations.
\end{proof}

\section{Lower Bound for Binary Subadditive Valuations}
\label{section:query-complexity-lb}
In this section we prove that, under binary subadditive valuations, an exponential number of value queries are required to obtain a sub-linear approximation for the Nash social welfare. 

\begin{restatable}{theorem}{queryComplexityBound}
\label{theorem:query-complexity}
For fair division instances $\langle [m], [n], \{ f_i \}_{i=1}^n \rangle$ with binary subadditive valuations and a fixed constant $\varepsilon \in (0,1]$, exponentially many value queries are necessarily required for finding an allocation with Nash social welfare at least $\frac{1}{n^{1-\varepsilon}}$ times the optimal. 
\end{restatable}
Towards establishing this theorem, we define two (families of) fair division instances, each with $n$ agents, $m=n^2$ goods, and binary subadditive valuations. In the first instance, all the agents will have the same binary subadditive valuation, $f : 2^{[m]} \mapsto \mathbb{R}_+$, while in the second instance, the valuations of the agents will be non-identical, $f'_i : 2^{[m]} \mapsto \mathbb{R}_+$ for each agent $i \in [n]$. In particular, we will construct the valuations, $f$ and $\{f'_i \}_i$, such that $(i)$ distinguishing whether the agents' valuations are $\{f'_i\}_i$ or $f$ requires an exponential number of value queries (Lemma \ref{proposition:exponential-queries}) and $(ii)$ the optimal Nash social welfare of the two instances differ multiplicatively by a linear factor. Since the second property implies that one can use any sub-linear approximation of the optimal Nash social welfare to distinguish between the two instances (i.e., between the two valuation settings), these properties will establish the stated query lower bound.

To specify the valuations, fix a small constant $\delta \in \left( 0, \frac{1}{16} \right)$ and write integers $\breakone \coloneqq \left\lfloor (1 + \delta) \ n^{4 \delta} \right\rfloor$ along with $\breaktwo \coloneqq \left\lfloor n^{1+2 \delta} \right\rfloor$. We will assume, throughout, that $n$ is large enough to ensure that the integers $\breakone, \breaktwo \in \mathbb{Z}_+$ satisfy $\breakone < \breaktwo$. With these parameters in hand, define valuation $f:2^{[m]} \mapsto \mathbb{Z}_+$ as follows

\begin{align*}
f(S) \coloneqq \begin{cases} 
      |S| & \text{if  $|S| \leq \breakone$,} \\
      \breakone & \text{if $\breakone < |S| \leq \breaktwo$,} \\
      \left\lceil \frac{\breakone \ |S|}{\breaktwo} \right\rceil & \text{otherwise, if  $|S| > \breaktwo$} 
   \end{cases}
\end{align*}

For constructing valuations $\{ f'_i \}_i$, consider a random $n$-partition, $T_1, T_2, \ldots, T_n$, of the set of goods $[m]$, with $|T_i| = n$ for each $i \in [n]$. Now, for every $i \in [n]$ and subset $S \subseteq [m]$, define $f'_i(S) \coloneqq \max \{ f(S), | S \cap T_i | \}$. The following two lemmas show that the constructed valuations are binary subadditive. 

\begin{lemma}
\label{proposition:binary-subadditivity-of-v}
The valuation $f$ (as defined above) is subadditive and has binary marginals. 
\end{lemma}
\begin{proof}
We will first prove that the function $f$ has binary marginals and then establish that it is subadditive. In particular, the following case analysis shows that for each subset $S \subseteq [m]$ and good $g \in [m] \setminus S$ we have $f(S \cup \{ g \}) - f(S) \in \{0,1\}$ \\
\noindent
{\it Case {\rm I}:} $|S| < \breakone$. In this case, the marginal increase is equal to one, $f(S \cup \{ g\}) - f(S) = |S+g| - |S| = 1$.

\noindent
{\it Case {\rm II}:} $|S| = \breakone$. Here, $\breakone +1 =   |S \cup \{ g \}| \leq \breaktwo$; recall that $\breakone + 1 \leq \breaktwo$. Therefore, the marginal increase is zero, $f(S \cup \{ g\}) - f(S) = \breakone - \breakone = 0$.

\noindent
{\it Case {\rm III}:} $\breakone < |S| < \breaktwo$. In this case, we have $\breakone < |S|+1 \leq \breaktwo$. Hence,  the marginal values are again zero, $f(S \cup \{g\}) - f(S) = \breakone - \breakone = 0$.

\noindent
{\it Case {\rm IV}:} $|S| = \breaktwo$. Since $|S|+1 > \breaktwo$, here we have 
\begin{align*}
f(S \cup \{ g \}) - f(S) & = \left\lceil \frac{\breakone \ (\breaktwo + 1)}{\breaktwo} \right\rceil - \breakone = \left\lceil \breakone + \frac{\breakone}{\breaktwo} \right\rceil -  \breakone \\
& = \left( \breakone + 1 \right) - \breakone \tag{since $\breakone$ is an integer and $\breakone/\breaktwo <1$} \\ 
& = 1.
\end{align*}

\noindent
{\it Case {\rm V}:} $|S| > \breaktwo$. In this case, the marginal increase is
\begin{align*}
f(S \cup \{ g\}) - f(S) & = \left\lceil \frac{\breakone \ (|S| + 1)}{\breaktwo} \right\rceil - \left\lceil \frac{\breakone \ |S|}{\breaktwo} \right\rceil = \left\lceil \frac{\breakone \ |S|}{\breaktwo} + \frac{\breakone}{\breaktwo} \right\rceil - \left\lceil \frac{\breakone \ |S|}{\breaktwo} \right\rceil \\
& \leq \left\lceil \frac{\breakone \ |S|}{\breaktwo} \right\rceil + \left\lceil \frac{\breakone}{\breaktwo} \right\rceil - \left\lceil \frac{\breakone \ |S|}{\breaktwo} \right\rceil \tag{since $\lceil a + b \rceil \leq \lceil a \rceil + \lceil b \rceil$ for any $a,b \in \mathbb{R}_+$} \\
& = \left\lceil \frac{\breakone}{\breaktwo} \right\rceil = 1 \tag{since $\breakone/\breaktwo <1$}
\end{align*}
By definition, $f$ is an integer-valued function and is monotonic. Therefore, the previous inequality ($f(S \cup \{ g \}) - f(S) \leq 1$) implies that the marginals are binary in this case as well.  
 
The five cases above establish that throughout the valuation $f$ has binary marginals. \\

We now establish the subadditivity of $f$, i.e., show that $f(S \cup T) \leq f(S) + f(T)$, for all subsets $S,T \subseteq [m]$. Towards this, we consider the following cases based on the cardinality of the set $S \cup T$ \\

\noindent
{\it Case {\rm I}:} $|S \cup T| \leq \breaktwo$. Note that $f(X) = \min \{ |X| , \breakone \}$, for any subset that satisfies $|X| \leq \breaktwo$. Hence, here we have 
\begin{align*}
f(S \cup T) & = \min \{ |S \cup T|, \breakone \} \\ 
& \leq \min \{ |S| + |T|,  \breakone \} \tag{since $|S \cup T| \leq |S| + |T|$} \\
& \leq \min\{|S|, \breakone\}  + \min\{|T|, \breakone\} \\
& = f(S) + f(T) \tag{since $|S| \leq \breaktwo$ and $|T| \leq \breaktwo$}
\end{align*}
Therefore, in this case $f$ is subadditive. \\

\noindent
{\it Case {\rm II}:} $|S \cup T| > \breaktwo$. Note that, for every $X \subseteq [m]$, the function $f$ (by definition) satisfies 
\begin{align}
\left\lceil |X| \ \frac{\breakone}{\breaktwo} \right\rceil \leq f(X)  \label{ineq:f-case} 
\end{align}
Furthermore, given that $|S \cup T| > \breaktwo$, the following equality holds  
\begin{align*}
f(S \cup T) & = \left\lceil |S \cup T| \ \frac{\breakone}{\breaktwo} \right\rceil \\ 
& \leq \left\lceil (|S| + |T|) \ \frac{\breakone}{\breaktwo} \right\rceil \tag{via $|S \cup T| \leq |S| + |T|$} \\ 
& \leq \left\lceil |S| \ \frac{\breakone}{\breaktwo} \right\rceil + \left\lceil |T| \ \frac{\breakone}{\breaktwo} \right\rceil \tag{since $\lceil a+b \rceil \leq \lceil a \rceil + \lceil b \rceil$ for all $a,b \in \mathbb{R}_+$} \\
& \leq f(S) + f(T) \tag{via (\ref{ineq:f-case})}
\end{align*}

Hence, this case analysis shows that $f$ is subadditive. Therefore, the lemma stands proved. 
\end{proof}

\begin{lemma}
\label{proposition:binary-subadditivity-of-f-prime}
The valuations $\{ f'_i \}_{i \in [n]}$ (as defined above) are subadditive and have binary marginals. 
\end{lemma}
\begin{proof}
For any subset $T \subseteq [m]$, consider the function $u_T(S) \coloneqq |S \cap T|$, for all subsets $S \subseteq [m]$. Since $u_T(\cdot)$ is an additive function with binary marginals, $u_T(\cdot)$ is binary subadditive. 

Recall that  $f'_i(S) = \max\{ f(S), |S \cap T_i| \} = \max\{ f(S), u_{T_i}(S) \}$, for every subset $S \subseteq [m]$. That is, $f'_i$ is defined to be the pointwise maximum of binary subadditive functions $f(\cdot)$ and $u_{T_i}(\cdot)$. 

Lemma \ref{proposition:subadditivity-max} (see Appendix \ref{appendix:query-complexity-lb}) shows that subadditivity is closed under the pointwise maximum operation. Lemma \ref{proposition:bin-marginals-max} (Appendix \ref{appendix:query-complexity-lb}) establishes an analogous result for the binary-marginals property. Hence, applying these two results, we obtain that $f'_i$ is a binary subadditive function, for each $i \in [n]$.  
\end{proof}

\begin{lemma}
\label{proposition:exponential-queries}
An exponential number of value queries are required to distinguish between the functions $f$ and $f'_i$, for any $i \in [n]$.
\end{lemma}
\begin{proof}
We will show that, for any subset $S \subseteq [m]$, the inequality $f'_i(S) \neq f(S)$ holds with exponentially small probability, over random subset $T_i$ (that define $f_i$).  Therefore, exponentially many value queries are required to distinguish between the two functions.  

Towards bounding the probability, we conduct the following case analysis based on the size of the set $S$:  \\
\noindent
{\it Case {\rm I}:} $|S| \leq \breakone$. In this case, we have $f(S) = |S| = f'_i(S)$. Hence, if the set $S$ is of size at most $\breakone$, then the functions $f$ and $f'_i$ are indistinguishable.\\

\noindent
{\it Case {\rm II}:} $\breakone < |S| \leq \breaktwo$. For any such subset $S \subseteq [m]$, that values $f'_i(S)$ and $f(S)$ differ iff $f'_i(S) = |S \cap T_i| > f(S) = \breakone$. That is, the queried value of set $S$ can be used to identity the underlying valuation iff $|S \cap T_i| \geq \breakone + 1 > (1+\delta)n^{4\delta}$.
Recall that $T_i \subset [m]$ is a random size-$n$ subset and $m = n^2$. Now, for each good $g \in S$, write $\chi_g$ to denote the indicator random variable that is equal to one iff $g \in T_i$, and note that 
\begin{align*}
\mathbb{E} \left[ |S \cap T_i | \right] = \sum_{g \in S} \mathbb{E} \left[ \chi_g \right] = |S| \ \frac{|T_i|}{m} = \frac{|S|}{n} \leq \frac{\breaktwo}{n} \leq n^{2 \delta} 
\end{align*}
Since the random variables $\chi_g$s are negatively associated \cite{BallsBinsDubhashiRanjan} and $(1+\delta)n^{4\delta} \geq n^{2 \delta} \geq \mathbb{E} \left[ |S \cap T_i | \right] $, the Chernoff bound gives us $\mathrm{Pr} \left\{|S\cap T_i| > (1+\delta)n^{4\delta} \right\} \leq 2^{-(1+\delta)n^{4\delta}}$; see, e.g., \cite[Theorem 4.4 (3)]{MU-Book}. Therefore, in this case one can distinguish between the valuations with an exponentially small probability only.  \\ 

\noindent
{\it Case {\rm III}:} $|S| > \breaktwo$. Analogous to the previous case, here the values $f(S)$ and $f'_i(S)$ differ iff $|S \cap T_i| > f(S) = \left\lceil \frac{\breakone \ |S|}{\breaktwo} \right\rceil$. That is, to differentiate we require 
\begin{align*}
|S \cap T_i| & \geq \frac{ \breakone \ |S|}{\breaktwo} \\
& \geq \frac{\left( (1+\delta) \ n^{4\delta} - (1+\delta) \right) \ |S|}{\breaktwo} \tag{since $ \breakone  \geq  (1+\delta) \ n^{4\delta} - 1 > (1+\delta) \ n^{4\delta} - (1+\delta)$} \\
& \geq \frac{\left( n^{4\delta} - 1 \right) (1+\delta) \ |S|}{n^{1+2\delta}} \tag{since $\breaktwo \leq n^{1+2\delta}$} \\
& = \left( n^{2\delta} - \frac{1}{ n^{2\delta}} \right) (1+\delta) \ \frac{|S|}{n} \\
& \geq (1+\delta) \ \frac{|S|}{n} \tag{for large $n$}\\
\end{align*} 
Furthermore, $\mathbb{E} \left[ |S \cap T_i | \right] = \frac{|S|}{n} \geq n^{2 \delta} $; here, the last inequality holds, since in the current case $|S| > \breaktwo =\left\lfloor n^{1+2 \delta} \right\rfloor$. Again, applying the Chernoff bound we get that the differentiating event occurs with exponentially small probability, $\mathrm{Pr} \left\{ |S\cap T_i| \geq (1+\delta)\frac{|S|}{n} \right\} \leq \mathrm{exp} \left( -\frac{n^{4 \delta} \ \delta^2}{3} \right)$. \\

Overall, the analysis shows that, to distinguish $f$ and $f'_i$, one necessarily needs to query the values of exponentially many subsets $S$. This completes the proof.  
\end{proof}

\subsection{Proof of Theorem \ref{theorem:query-complexity}}

Here, we establish Theorem \ref{theorem:query-complexity}, our main negative result for binary subadditive valuations. 

With $n$ agents and $m=n^2$ goods, we consider two families of instances with binary subadditive valuations (see Lemmas \ref{proposition:binary-subadditivity-of-v} and \ref{proposition:binary-subadditivity-of-f-prime}): the first one in which all the agents have the same valuation $f$, and the other wherein the agents' valuations are $\{f'_i \}_{i=1}^n$.   Lemma \ref{proposition:exponential-queries} shows that exponentially many value value queries are required to distinguish between these two cases, i.e., to determine whether the agents' valuations are $f$ or $\{f'_i\}_i$.  

We will next establish that such a distinction can be made via an $n^{1-\varepsilon}$ approximation to the optimal Nash social welfare and, hence, obtain the stated query complexity of approximating the Nash social welfare. Note that,  under valuations $\{f'_i\}_i$, the optimal Nash welfare is equal to $n$. In particular, allocating bundle $T_i$ to agent $i$ leads to $f'_i(T_i) = n$, for each $i \in [n]$, i.e., here the Nash social welfare of the allocation $(T_1, \ldots, T_n)$ is $n$.  

By contrast, under valuation $f$, the optimal Nash social welfare is at most $ \left(2 n^{4\delta} + 1\right)$. In fact, the following argument shows that, for any allocation $(A_1, \ldots, A_n)$, the average social welfare $\frac{1}{n} \sum_{i=1}^n f(A_i) \leq 2 n^{4\delta} + 1$. Hence, via the AM-GM inequality, this upper bound holds for the optimal Nash social welfare as well.  

Let $\alloc = (A_1, A_2, \ldots, A_n)$ be the allocation that maximizes the average social welfare under $f$. We can assume, without loss of generality, that for each agent $i \in [n]$, either $|A_i| > \breaktwo$ or $|A_i| \leq \breakone$. Otherwise, if for some agent $j \in [n]$, we have $\breakone < |A_j| \leq \breaktwo$, then we can iteratively remove goods from $A_j$ until $|A_j| = \breakone$. This update will not decrease $f(A_j)$ (this value will continue to be $\breakone$) and, hence, the social welfare remains unchanged as well. For ease of analysis, we further modify allocation $\alloc$: while there are two agents $j, k \in [n]$ with $|A_k| \geq |A_j| > \breaktwo$, we iteratively transfer goods from $A_j$ to $A_k$ until $|A_j| = \breakone$. Note that after each transfer the social welfare decreases by at most one (i.e., the drop in average social welfare is at most ${1}/{n}$):\footnote{Recall that for any $a, b \in \mathbb{R}_+$, the following inequalities hold: $\lceil a \rceil + \lceil b \rceil - 1 \leq \lceil a + b \rceil \leq \lceil a \rceil + \lceil b \rceil$.} for any $s \leq |A_j|$, we have $\left\lceil \frac{\breakone}{\breaktwo} \left(  |A_k| + s \right) \right\rceil + \left\lceil \frac{\breakone}{\breaktwo} \left(  |A_j|  - s \right) \right\rceil \geq \left\lceil \frac{\breakone}{\breaktwo} \left(  |A_k| + s \right)  + \frac{\breakone}{\breaktwo} \left(  |A_j|  - s \right) \right\rceil = \left\lceil \frac{\breakone}{\breaktwo}  |A_k|  + \frac{\breakone}{\breaktwo} |A_j| \right\rceil \geq \left\lceil \frac{\breakone}{\breaktwo}  |A_k| \right\rceil + \left\lceil \frac{\breakone}{\breaktwo} |A_j| \right\rceil - 1$.   

Since at most $n$ such transfers can occur (between pairs of agents), the \emph{average} social welfare of allocation $\alloc$ decreases by at most one after all the transfers. Now, allocation $\alloc$ has exactly one agent with bundle size greater than $\breakone$.  This observation gives us the following upper bound
\begin{align*}
\frac{1}{n} \sum_{i=1}^n f(A_i) & \leq \frac{1}{n} \left( (n-1) \breakone + \left\lceil \frac{\breakone}{\breaktwo} m \right\rceil \right) \\
& \leq \frac{1}{n} \left( n \breakone +  \frac{\breakone}{\breaktwo} m \right) \\
& \leq \breakone + \frac{n \breakone}{\breaktwo} \tag{since $m=n^2$} \\
& \leq 2 n^{4 \delta} \tag{since $\breakone =\lfloor (1 + \delta) n^{4 \delta} \rfloor$ and $\breaktwo =\lfloor n^{1+2\delta} \rfloor$}
\end{align*}
 
Therefore, under valuation $f$, the average social welfare is at most $ 2 n^{4 \delta} + 1$. As mentioned previously, this implies that the ratio of the optimal Nash welfares under $f'_i$s and $f$ is $\Omega(n^{1- 4 \delta})$. This, overall, establishes that any sub-linear approximation would differentiate between the valuations and, hence, require exponentially many value queries. The theorem stands proved. 

\section{Conclusion and Future Work}
We develop algorithmic and hardness result for Nash social welfare maximization under binary $\XOS$ and binary subadditive valuations. Our algorithm provides (under binary $\XOS$ valuations) constant-factor approximations simultaneously for Nash social welfare, social welfare, and $\GMMS$. It would be interesting to extend the positive result for Nash social welfare to the asymmetric version, wherein each agent has an associated weight (entitlement) $e_i \in \mathbb{R}_+$, and the objective is to find an allocation $(X_1, \ldots, X_n)$ that maximizes $\left( \prod_i \left( v_i(X_i) \right)^{e_i} \right)^{\frac{1}{\sum_i e_i}}$. Another interesting direction for future work is to develop, under binary $\XOS$ valuations, constant-factor approximation algorithms for $p$-mean welfare maximization, with $p \leq 1$. 

\section*{Acknowledgements}
Siddharth Barman gratefully acknowledges the support of a Ramanujan Fellowship (SERB - {SB/S2/RJN-128/2015}).

\bibliographystyle{alpha} 
\bibliography{bin-marginals-refs}

\newcommand{\etalchar}[1]{$^{#1}$}
\begin{thebibliography}{AGMV18}

\bibitem[AGMV18]{AGM+18nash}
Nima Anari, Shayan~Oveis Gharan, Tung Mai, and Vijay~V Vazirani.
\newblock {Nash Social Welfare for Indivisible Items under Separable,
  Piecewise-Linear Concave Utilities}.
\newblock In {\em Proceedings of the 29th Annual {ACM}-{SIAM} Symposium on
  Discrete Algorithms (SODA)}, pages 2274--2290, 2018.

\bibitem[AGSS17]{AGS+17nash}
Nima Anari, Shayan~Oveis Gharan, Amin Saberi, and Mohit Singh.
\newblock {Nash Social Welfare, Matrix Permanent, and Stable Polynomials}.
\newblock In {\em Proceedings of the 8th Conference on Innovations in
  Theoretical Computer Science (ITCS)}, 2017.

\bibitem[BBKN18]{BBKN18}
Siddharth Barman, Arpita Biswas, Sanath Krishnamurthy, and Yadati Narahari.
\newblock Groupwise maximin fair allocation of indivisible goods.
\newblock In {\em Proceedings of the AAAI Conference on Artificial
  Intelligence}, volume~32, 2018.

\bibitem[BBKS20]{BBKS20}
Siddharth Barman, Umang Bhaskar, Anand Krishna, and Ranjani~G Sundaram.
\newblock Tight approximation algorithms for p-mean welfare under subadditive
  valuations.
\newblock In {\em 28th Annual European Symposium on Algorithms (ESA 2020)}.
  Schloss Dagstuhl-Leibniz-Zentrum f{\"u}r Informatik, 2020.

\bibitem[BCIZ20]{benabbou2020finding}
Nawal Benabbou, Mithun Chakraborty, Ayumi Igarashi, and Yair Zick.
\newblock Finding fair and efficient allocations when valuations don’t add
  up.
\newblock In {\em International Symposium on Algorithmic Game Theory}, pages
  32--46. Springer, 2020.

\bibitem[BEF21]{BEF2020}
Moshe Babaioff, Tomer Ezra, and Uriel Feige.
\newblock Fair and truthful mechanisms for dichotomous valuations.
\newblock In {\em Proceedings of the AAAI Conference on Artificial
  Intelligence}, volume~35. AAAI, 2021.

\bibitem[BGHM17]{BGH+17earning}
Xiaohui Bei, Jugal Garg, Martin Hoefer, and Kurt Mehlhorn.
\newblock {Earning Limits in Fisher Markets with Spending-Constraint
  Utilities}.
\newblock In {\em Proceedings of the International Symposium on Algorithmic
  Game Theory (SAGT)}, pages 67--79, 2017.

\bibitem[BKV18a]{BKV18}
Siddharth Barman, Sanath~Kumar Krishnamurthy, and Rohit Vaish.
\newblock Finding fair and efficient allocations.
\newblock In {\em Proceedings of the 2018 ACM Conference on Economics and
  Computation}, pages 557--574, 2018.

\bibitem[BKV18b]{barman2018greedy}
Siddharth Barman, Sanath~Kumar Krishnamurthy, and Rohit Vaish.
\newblock Greedy algorithms for maximizing {N}ash social welfare.
\newblock In {\em Proceedings of the 17th International Conference on
  Autonomous Agents and MultiAgent Systems}, pages 7--13, 2018.

\bibitem[BL08]{BL08}
Sylvain Bouveret and J{\'e}r{\^o}me Lang.
\newblock Efficiency and envy-freeness in fair division of indivisible goods:
  Logical representation and complexity.
\newblock {\em Journal of Artificial Intelligence Research}, 32:525--564, 2008.

\bibitem[BL16]{BL16}
Sylvain Bouveret and Michel Lema{\^\i}tre.
\newblock Characterizing conflicts in fair division of indivisible goods using
  a scale of criteria.
\newblock {\em Autonomous Agents and Multi-Agent Systems}, 30(2):259--290,
  2016.

\bibitem[BMS05]{BMS05}
Anna Bogomolnaia, Herv{\'e} Moulin, and Richard Stong.
\newblock Collective choice under dichotomous preferences.
\newblock {\em Journal of Economic Theory}, 122(2):165--184, 2005.

\bibitem[Bud11]{Bud11}
Eric Budish.
\newblock The combinatorial assignment problem: Approximate competitive
  equilibrium from equal incomes.
\newblock {\em Journal of Political Economy}, 119(6):1061--1103, 2011.

\bibitem[BV21]{BV21}
Siddharth Barman and Paritosh Verma.
\newblock Existence and computation of maximin fair allocations under
  matroid-rank valuations.
\newblock In {\em Proceedings of the 20th International Conference on
  Autonomous Agents and MultiAgent Systems}, AAMAS '21, pages 169--177, 2021.

\bibitem[CC03]{chlebik2003hardness}
Miroslav Chleb{\'\i}k and Janka Chleb{\'\i}kov{\'a}.
\newblock Inapproximability results for bounded variants of optimization
  problems.
\newblock In {\em International Symposium on Fundamentals of Computation
  Theory}, pages 27--38. Springer, 2003.

\bibitem[CDG{\etalchar{+}}17]{CDGJ+17}
Richard Cole, Nikhil Devanur, Vasilis Gkatzelis, Kamal Jain, Tung Mai, Vijay~V
  Vazirani, and Sadra Yazdanbod.
\newblock Convex program duality, {F}isher markets, and {N}ash social welfare.
\newblock In {\em Proceedings of the 2017 ACM Conference on Economics and
  Computation}, pages 459--460, 2017.

\bibitem[CG15]{CG15}
Richard Cole and Vasilis Gkatzelis.
\newblock Approximating the {N}ash social welfare with indivisible items.
\newblock In {\em Proceedings of the forty-seventh annual ACM symposium on
  Theory of computing}, pages 371--380, 2015.

\bibitem[CGM21]{chaudhury2020fair}
Bhaskar~Ray Chaudhury, Jugal Garg, and Ruta Mehta.
\newblock Fair and efficient allocations under subadditive valuations.
\newblock In {\em Proceedings of the AAAI Conference on Artificial
  Intelligence}, volume~35. AAAI, 2021.

\bibitem[DNS10]{DobzinskiNS10}
Shahar Dobzinski, Noam Nisan, and Michael Schapira.
\newblock Approximation algorithms for combinatorial auctions with
  complement-free bidders.
\newblock {\em Mathematics of Operations Research}, 35(1):1--13, 2010.

\bibitem[DR96]{BallsBinsDubhashiRanjan}
Devdatt~P Dubhashi and Desh Ranjan.
\newblock Balls and bins: A study in negative dependence.
\newblock {\em BRICS Report Series}, 3(25), 1996.

\bibitem[DS15]{DS15}
Andreas Darmann and Joachim Schauer.
\newblock Maximizing {N}ash product social welfare in allocating indivisible
  goods.
\newblock {\em European Journal of Operational Research}, 247(2):548--559,
  2015.

\bibitem[Fei06]{Feige06}
Uriel Feige.
\newblock On maximizing welfare when utility functions are subadditive.
\newblock In {\em Proceedings of the Thirty-Eighth Annual ACM Symposium on
  Theory of Computing}, STOC '06, pages 41--50, New York, NY, USA, 2006.
  Association for Computing Machinery.

\bibitem[Fre10]{F10}
Guilherme Freitas.
\newblock Combinatorial assignment under dichotomous preferences, 2010.

\bibitem[GHM18]{GHM18}
Jugal Garg, Martin Hoefer, and Kurt Mehlhorn.
\newblock Approximating the {N}ash social welfare with budget-additive
  valuations.
\newblock In {\em Proceedings of the Twenty-Ninth Annual ACM-SIAM Symposium on
  Discrete Algorithms}, pages 2326--2340. SIAM, 2018.

\bibitem[GHS{\etalchar{+}}18]{ghodsi2018fair}
Mohammad Ghodsi, MohammadTaghi HajiAghayi, Masoud Seddighin, Saeed Seddighin,
  and Hadi Yami.
\newblock Fair allocation of indivisible goods: Improvements and
  generalizations.
\newblock In {\em Proceedings of the 2018 ACM Conference on Economics and
  Computation}, pages 539--556, 2018.

\bibitem[GHV21]{garg2021approximating}
Jugal Garg, Edin Husi{\'c}, and L{\'a}szl{\'o}~A V{\'e}gh.
\newblock {Approximating Nash social welfare under Rado valuations}.
\newblock In {\em Proceedings of the 53rd Annual ACM SIGACT Symposium on Theory
  of Computing}, pages 1412--1425, 2021.

\bibitem[GKK20]{GargKK20}
Jugal Garg, Pooja Kulkarni, and Rucha Kulkarni.
\newblock Approximating {N}ash social welfare under submodular valuations
  through (un)matchings.
\newblock In {\em Proceedings of the fourteenth annual ACM-SIAM symposium on
  discrete algorithms}, pages 2673--2687. SIAM, 2020.

\bibitem[GP15]{GP15}
Jonathan Goldman and Ariel~D Procaccia.
\newblock Spliddit: Unleashing fair division algorithms.
\newblock {\em ACM SIGecom Exchanges}, 13(2):41--46, 2015.

\bibitem[GT20]{GT20}
Jugal Garg and Setareh Taki.
\newblock An improved approximation algorithm for maximin shares.
\newblock In {\em Proceedings of the 21st ACM Conference on Economics and
  Computation}, pages 379--380, 2020.

\bibitem[HPPS20]{halpern2020fair}
Daniel Halpern, Ariel~D Procaccia, Alexandros Psomas, and Nisarg Shah.
\newblock Fair division with binary valuations: One rule to rule them all.
\newblock In {\em International Conference on Web and Internet Economics},
  pages 370--383. Springer, 2020.

\bibitem[KPS18]{KPS18}
David Kurokawa, Ariel~D Procaccia, and Nisarg Shah.
\newblock Leximin allocations in the real world.
\newblock {\em ACM Transactions on Economics and Computation (TEAC)},
  6(3-4):1--24, 2018.

\bibitem[KPW16]{KPW16}
David Kurokawa, Ariel~D Procaccia, and Junxing Wang.
\newblock When can the maximin share guarantee be guaranteed?
\newblock In {\em Proceedings of the Thirtieth AAAI Conference on Artificial
  Intelligence}, pages 523--529, 2016.

\bibitem[LBMS17]{LMS17}
Kevin Leyton-Brown, Paul Milgrom, and Ilya Segal.
\newblock Economics and computer science of a radio spectrum reallocation.
\newblock {\em Proceedings of the National Academy of Sciences},
  114(28):7202--7209, 2017.

\bibitem[Lee17]{Lee17}
Euiwoong Lee.
\newblock Apx-hardness of maximizing {N}ash social welfare with indivisible
  items.
\newblock {\em Information Processing Letters}, 122:17--20, 2017.

\bibitem[LV21a]{LV21}
Wenzheng Li and Jan Vondr{\'a}k.
\newblock A constant-factor approximation algorithm for {N}ash social welfare
  with submodular valuations.
\newblock {\em arXiv preprint arXiv:2103.10536}, 2021.

\bibitem[LV21b]{li2021fair}
Zhentao Li and Adrian Vetta.
\newblock The fair division of hereditary set systems.
\newblock {\em ACM Transactions on Economics and Computation (TEAC)},
  9(2):1--19, 2021.

\bibitem[Mou04]{Moul03}
Herv{\'e} Moulin.
\newblock {\em Fair division and collective welfare}.
\newblock MIT press, 2004.

\bibitem[MU17]{MU-Book}
Michael Mitzenmacher and Eli Upfal.
\newblock {\em Probability and computing: Randomization and probabilistic
  techniques in algorithms and data analysis}.
\newblock Cambridge university press, 2017.

\bibitem[Nas50]{Nash50}
John Nash.
\newblock The bargaining problem.
\newblock {\em Econometrica}, 18(2):155--162, 1950.

\bibitem[NRTV07]{nisan2007algorithmic}
Noam Nisan, Tim Roughgarden, Eva Tardos, and Vijay~V Vazirani.
\newblock {\em Algorithmic game theory}.
\newblock Cambridge university press, 2007.

\bibitem[Ort20]{O20}
Josu{\'e} Ortega.
\newblock Multi-unit assignment under dichotomous preferences.
\newblock {\em Mathematical Social Sciences}, 103:15--24, 2020.

\bibitem[PW14]{PW14}
Ariel~D Procaccia and Junxing Wang.
\newblock Fair enough: Guaranteeing approximate maximin shares.
\newblock In {\em Proceedings of the fifteenth ACM conference on Economics and
  computation}, pages 675--692, 2014.

\bibitem[RS{\"U}05]{RSU05}
Alvin~E Roth, Tayfun S{\"o}nmez, and M~Utku {\"U}nver.
\newblock Pairwise kidney exchange.
\newblock {\em Journal of Economic theory}, 125(2):151--188, 2005.

\bibitem[Sch03]{schrijver2003combinatorial}
Alexander Schrijver.
\newblock {\em Combinatorial optimization: polyhedra and efficiency},
  volume~24.
\newblock Springer Science \& Business Media, 2003.

\end{thebibliography}

\appendix
\section{Missing Proofs from Section \ref{section:prelims}}
\label{appendix:prelims}

\PropNonWasteful*
\begin{proof} 
Consider any subset $X \subset S$ of the non-wasteful set $S$ and, towards a contradiction, assume that $v(X) < |X|$. Note that $v(X) \leq |X|$, 
since function $v$ has binary marginals. Now, consider the process of adding goods $\overline{g} \in S \setminus X$ into the set $X$ one by one. Again, the fact that $v$ has binary marginals ensures that in this process the increase in value after the addition of each good $\overline{g}$ is at most one. Therefore, at the end of the process $v(X \cup (S \setminus X)) \leq v(X) + |S \setminus X| < |X| + |S \setminus X| = |S|$. This inequality, $v(S) < |S|$, contradicts the fact that $S$ is non-wasteful and completes the proof. 
\end{proof}

We next restate and prove Theorem \ref{theorem:xos-definitions}. 
\ThmBinaryXOSChar*
\begin{proof}
We will establish the equivalence between these four properties of binary $\XOS$ functions by proving the following cyclic implications $(\mathrm{P}_1) \implies (\mathrm{P}_2) \implies (\mathrm{P}_3) \implies (\mathrm{P}_4) \implies (\mathrm{P}_1)$. \\

\noindent
{\it Implication {\rm I}:} ($\mathrm{P}_1$)$\implies$($\mathrm{P}_2$). Let $v$ be a function that is $\XOS$ and has binary marginals. The binary-marginals property ensures that $v$ is monotonic as well, i.e., $v(X) \leq v(Y)$ for any subsets $X \subseteq Y \subseteq [m]$. To establish property ($\mathrm{P}_2$), consider the following iterative procedure for any subset $S \subseteq [m]$: initialize $S' = S$ and while there exists a good $g' \in S'$ such that $v(S' \setminus \{ g' \}) = v(S)$, we update $S' \leftarrow S' \setminus \{ g' \}$. Let $X \subseteq S$ be the subset that remains at this end of this procedure, i.e., $v(S) = v(X)$ and, for every good $g \in X$, we have $v\left( X \setminus \{g\} \right) = v(X) - 1$; note that  $v(X) - v\left( X \setminus \{g\} \right) \in \{0,1\}$, since $v$ has binary marginals. Since the function $v$ is $\XOS$, we have that $v(X) = \ell(X)$ for some additive function $\ell: 2^{[m]} \mapsto \mathbb{R}_+$, i.e., $v(X) = \ell(X) = \sum_{g \in X} \ell(g)$. 

We will now show that $\ell(g) = 1$ for every $g \in X$. This will imply that property ($\mathrm{P}_2$) holds, since we will then have $v(S) = v(X) = \sum_{g \in X} \ell(g) = |X|$.  
Towards a contradiction, assume that $\ell(\widehat{g}) \neq 1$ for some good $\widehat{g} \in X$. In such a case we have $\ell(\widehat{g}) < 1$.\footnote{Note that it cannot be the case that $\ell(\widehat{g}) > 1$, since the binary-marginals property ensures $\ell( \widehat{g}) \leq v(\widehat{g}) \leq 1$.} Therefore, 

\begin{align*}
v \left( X \setminus \left\{ \widehat{g} \right\} \right) & \geq \ell \left( X \setminus \left\{\widehat{g} \right\} \right) \tag{since $v$ is $\XOS$}\\
& = \ell(X) - \ell(\widehat{g}) \tag{function $\ell$ is additive}\\
& = v(X) - \ell(\widehat{g}) \\ 
& > v(X) - 1. \tag{since $\ell(\widehat{g}) < 1$}
\end{align*}
This (strict) inequality, $v \left( X \setminus \left\{ \widehat{g} \right\} \right) > v(X) - 1$, contradicts the fact that, for every good $g \in X$, we have $v(X \setminus\{g\}) = v(X) - 1$. Therefore, $\ell(g) = 1$ for every good $g \in X$, and property ($\mathrm{P}_2$) follows. \\

\noindent
{\it Implication {\rm II}:} ($\mathrm{P}_2$) $\implies$ ($\mathrm{P}_3$). To show that a function $v$ that satisfies property ($\mathrm{P}_2$) also satisfies ($\mathrm{P}_3$), we will construct a family of additive functions $\{\ell_t\}_t$ such that (i) for every set $S \subseteq [m]$, we have $v(S) = \max_{t} \ell_t(S)$, and (ii) the additive functions $\{\ell_t\}_t$ have binary marginals. 

To construct such a collection of additive functions, consider the collection of all non-wasteful subsets (with respect to $v$). In particular, write $\mathcal{T} \coloneqq \{ T \subseteq [m] : v(T) = |T | \}$, and, for each subset $T \in \mathcal{T}$, define additive function $\ell_T(X) \coloneqq |X \cap T|$ (for every subset $X \subseteq [m]$). Indeed, the additive functions $\left\{ \ell_T \right\}_{T \in \mathcal{T}}$ have binary marginals. It remains to show that these functions also satisfy requirement (i). Towards this, for any set $S \subseteq [m]$, let $X \subseteq S$ be the subset that satisfies ($\mathrm{P}_2$), i.e., $v(S) = v(X) = |X|$. Since $X$ is non-wasteful, we have $X \in \mathcal{T}$ and $v(S) = |X| = |S \cap X| = \ell_X(S) \leq \max_{T \in \mathcal{T} } \ \ell_T(X)$. 

Furthermore, the fact that function $v$ has binary marginals (see ($\mathrm{P}_2$)) implies that $v$ is monotonic. Therefore, for \emph{any} non-wasteful subset $T$, we have 
\begin{align*}
v(S) & \geq  v(S \cap T) \tag{$v$ is monotonic} \\
& = |S \cap T| \tag{applying Proposition \ref{proposition:non-wasteful-subsets} to $(S \cap T) \subseteq T \in \mathcal{T}$ } \\
& \geq \max_{T' \in \mathcal{T}} |S \cap T'|. \tag{maximizing over all $T' \in \mathcal{T}$} \\
& = \max_{T' \in \mathcal{T}} \ell_{T'} (S)
\end{align*}
Therefore, $v(S) \geq \max_{T \in \mathcal{T}} \ \ell_T(S)$ and property ($\mathrm{P}_3$) holds. \\

\noindent
{\it Implication {\rm III}:} ($\mathrm{P}_3$) $\implies$ ($\mathrm{P}_4$). Here, the support of the binary additive functions $\{ \ell_t \}_{t=1}^L$ that define $v$ in ($\mathrm{P}_3$)  give us the set family $\mathcal{F}$  for ($\mathrm{P}_4$); specifically, write $F_t \coloneqq \left\{ g \in [m] : \ell_t(g) = 1 \right\}$, for each $t \in [L]$, and $\mathcal{F} \coloneqq \left\{  F_t \right\}_{t=1}^L$. Since $v(S) = \max_t  \ell_t(S) = \max_t |F_t \cap S| = \max_{ F \in \mathcal{F} } |F \cap S|$, for any set $S \subseteq [m]$, property ($\mathrm{P}_4$) holds. \\

\noindent
{\it Implication {\rm IV}:} ($\mathrm{P}_4$) $\implies$ ($\mathrm{P}_1$). Complementary to the previous implication, we obtain a collection of additive functions $\{ \ell_F \}_{F \in \mathcal{F}}$ from the set family $\mathcal{F}$ (that defines function $v$ in ($\mathrm{P}_4$)). In particular, for each $F \in \mathcal{F}$, define $\ell_F(X) = |F \cap X|$ (for every $X \subseteq [m]$). Note that the functions $\ell_F(\cdot)$s are binary additive and $v(S) = \max_{F \in \mathcal{F}} \ |F \cap S| = \max_{F \in \mathcal{F}} \ell_F(S)$. Given that function $v$ can be expressed a pointwise maximum of additive functions, we get that $v$ is $\XOS$. Also, $v$ has binary marginals--this follows from the observation that each $\ell_F(\cdot)$ has binary marginals and this property is preserved under pointwise maximization (Lemma \ref{proposition:bin-marginals-max}). Therefore, $v$ satisfies ($\mathrm{P}_1$).\\

Overall, we obtain the equivalence among the four properties. 
\end{proof}

Next, we restate and prove Lemma \ref{proposition:non-wasteful-allocation}.

\PropNonWastefulAlloc*
\begin{proof}
For each agent $i \in [n]$, the valuation, $v_i$, is a binary $\XOS$ function. Therefore, property $(P_2)$ in Theorem \ref{theorem:xos-definitions} ensures that for each given bundle $P_i$ there necessarily exists a non-wasteful subset $P'_i \subseteq P_i$ with the property that $|P'_i| = v_i(P'_i) = v_i(P_i)$. This establishes the existence of the desired allocation $\palloc' = (P'_1, P'_2, \ldots, P'_n)$. 
\end{proof}

\section{Missing Proofs from Section \ref{section:nsw}}
\label{appendix:nsw}

\begin{proposition}
\label{prop:constrained-packing}
Let $D_0, D_1, \ldots, D_\ell$ be a collection of pairwise disjoint subsets of goods such that $|D_k|$ is an integer multiple of $4^{k+2}$, for each $0 \leq k \leq \ell$. Also, let $\mathcal{B} = (B_1, B_2, \ldots , B_t)$ be any $t$-partition of the set $\cup_{k=0}^\ell D_k$ with the property that \\
(P): For any index $k$ and each good $g \in D_k$, if $g \in B_b$, then $|B_b| \leq 4^{k+2}$ (i.e., goods in $D_k$ must be assigned among bundles of size at most $4^{k+2}$).\\
Then, in the partition, the number of bundles $t \geq \sum_{k=0}^\ell \frac{|D_k|}{4^{k+2}}$.
\end{proposition}
\begin{proof}
Index the bundles in the partition $\mathcal{B}$ such that $|B_1| \leq |B_2| \leq \ldots \leq |B_t|$. We will show that, with this indexing in hand, one can further assume that goods from lower indexed $D_k$s are assigned to lower indexed bundles; specifically, the partition $\mathcal{B}$ satisfies the following property \\
(Q): With index $a < b$, if bundle $B_\alpha$ contains a good from $D_a$ (i.e., $D_a \cap B_\alpha \neq \emptyset$) and bundle $B_\beta$ contains a good from $D_b$, then $\alpha \leq \beta$. 

In particular, this property ensures that all the goods in $D_0$ are assigned to the lowest indexed bundles, the subsequent bundles contain goods from $D_1$, and so on. 

We can directly count the smallest number the bundles, say $t^*$, that are required to cover $\cup_{k=0}^\ell D_k$ while satisfying properties (P) and (Q): property (Q) ensures that for some $t^*_0 \in \mathbb{Z}_+$, the first $t^*_0$ bundles contain all the goods in $D_0$. Since $|D_0|$ is an integer multiple of $4^2$ and each good in $D_0$ must be contained in a bundle of of size at most $4^2$ (property (P)), we get that $t^*_0 = |D_0|/4^2$. Inductively, for $D_k$, we require $t^*_k = |D_k|/4^{k+2}$ bundles.  Therefore, the desired bound follows $t \geq t^* = \sum_{k=0}^\ell t^*_k = \sum_{k=0}^\ell \frac{|D_k|}{4^{k+2}}$.

To complete the proof, we will now show that any partition $\mathcal{B} = (B_1, B_2, \ldots, B_t)$ can be iteratively transformed into another partition---without increasing $t$--that satisfies (Q), as well as (P). 

Assume that allocation $\mathcal{B}$ does not satisfy (Q) for bundles $B_\alpha $ and $B_\beta$, i.e., we have $\alpha > \beta$. In such a case, we can swap goods $g' \in B_\alpha \cap D_a$ and $g'' \in B_\beta \cap D_b$ between the two bundles and continue to maintain (P): since property (P) was satisfied before the swap, we know that $|B_\alpha| \leq 4^{a+2}$, and $|B_\beta| \leq 4^{b+2}$. Furthermore, here $a<b$ and $\alpha > \beta$. Hence, after the swap, (P) is satisfied for good $g'$: $|B_\beta| \leq |B_\alpha| \leq 4^{a+2}$. Property (P) is satisfied for $g''$ as well: $|B_\alpha| \leq 4^{a+2} < 4^{b+2}$. For all other goods, (P) directly continues to hold. 

Note that only a polynomial number of such swaps can be performed before property (Q) is satisfied. Hence, a repeated application of this swapping operation provides a partition that has the same number of bundles $t$, and it satisfies both properties (P) and (Q). This completes the proof.
\end{proof}

\begin{proposition}
\label{prop:power-product}
For any integer $\ell \geq 1$, we have $\prod_{k=0}^{\ell} \Big( \frac{1}{2^{k+1}} \Big)^{\frac{1}{2^k}} \geq \frac{1}{16}$.
\end{proposition}
\begin{proof}
Considering the logarithm (to the base $2$) of the left-hand-side of the stated inequality, we get  
\begin{align*}
\log  \left( \prod_{k=0}^{\ell} \left( \frac{1}{2^{k+1}}  \right)^{\frac{1}{2^k}}  \right) 
= \sum_{k=0}^\ell \frac{\log  (1/2^{k+1})}{2^k}
= \sum_{k=0}^\ell \frac{-(k+1)}{2^k}
\end{align*}
Hence, to prove the proposition, it is suffices to show that 
\begin{align}
\sum_{k=0}^\ell \frac{-(k+1)}{2^k} \geq \log \left(\frac{1}{16} \right) = -4  \label{equa-0} 
\end{align}

Towards this, consider the infinite sum 
\begin{equation}
\label{equa-1}
S \coloneqq \sum_{k=0}^{\infty} \frac{-(k+1)}{2^k}.
\end{equation}

Multiplying by $2$ on both sides we get 
\begin{equation}
\label{equa-2}
2S = \sum_{k=0}^{\infty} \frac{-(k+1)}{2^{k-1}} 
= -2 + \sum_{k=1}^{\infty} \frac{-(k+1)}{2^{k-1}}
= -2 + \sum_{k=0}^{\infty} \frac{-(k+2)}{2^k}
\end{equation}

Subtracting equation (\ref{equa-1}) from (\ref{equa-2}) gives us 

\begin{equation}
S = -2 + \sum_{k=0}^{\infty} \frac{(k+1)-(k+2)}{2^k} = -2 + \sum_{k=0}^{\infty} \frac{-1}{2^k} = -4 \label{ineq:equa-some}
\end{equation}
This completes the proof as equation (\ref{ineq:equa-some}) establishes the desired inequality (\ref{equa-0})

\begin{equation}
\sum_{k=0}^{\ell} \frac{-(k+1)}{2^k} \geq \sum_{k=0}^{\infty} \frac{-(k+1)}{2^k} = S = -4
\end{equation}
\end{proof}

\section{Missing Proofs from Section \ref{section:query-complexity-lb}}
\label{appendix:query-complexity-lb}

\begin{lemma}
\label{proposition:subadditivity-max}
Let $v, u: 2^{[m]} \mapsto \mathbb{R}_+$ be subadditive functions, and let $w(X) \coloneqq \max \{v(X), u(X) \}$, for each subset $X \subseteq [m]$. Then, the function $w: 2^{[m]} \mapsto \mathbb{R}_+$ is subadditive as well.
\end{lemma}
\begin{proof}
For any two subsets $S,T \subseteq [m]$, by definition of $w(\cdot)$ we have 
\begin{align*}
w(S \cup T) & = \max \{ v(S \cup T), u(S \cup T) \} \\
& \leq \max \{ v(S) + v(T), u(S) + u(T) \} \tag{$v$ and $u$ are subadditive} \\
& \leq \max \{ v(S) , u(S) \} + \max \{ v(T) , u(T) \} \tag{regrouping terms}  \\
& = w(S) + w(T)
\end{align*}
This proves that, as stated, the function $w$ is subadditive.
\end{proof}

\begin{lemma}
\label{proposition:bin-marginals-max}
If functions $v, u : 2^{[m]} \mapsto \mathbb{R}_+$ have binary marginals, then the function $w : 2^{[m]} \mapsto \mathbb{R}_+$ defined as $w(X) \coloneqq  \max \{ v(X), u(X) \}$ (for each $X \subseteq [m]$) also has binary marginals.
\end{lemma}
\begin{proof}
We will prove that, for any subset $S \subset [m]$ and good $g \in [m] \setminus S$, the marginal $w(S \cup \{g\}) - w(S) \in \{0,1\}$. Note that, since the functions $v$ and $u$ have 
binary marginals, they are integer-valued, i.e., $v(X), u(X) \in \mathbb{Z}_+$ for all $X \subseteq [m]$. Hence, $w$ is integer-valued as well, and to obtain the binary-marginals property for this function it suffices to show that $0 \leq w(S \cup \{g\}) - w(S) \leq 1$.

First, we will show that the marginal, $w(S \cup \{g\}) - w(S)$, is at least zero. Towards this, note that functions $v$ and $u$ are monotonic, since their marginals are binary. Hence, function $w$ is monotonic as well: $w(S \cup \{g\}) = \max \{ v(S \cup \{g\}), u (S\cup \{g\}) \} \geq \max \{ v(S), u (S) \} = w(S)$. That is,  $w(S \cup \{g\}) - w(S) \geq 0$. 

Finally, we show that marginal, $w(S \cup \{g\}) - w(S)$, is at most one. Write $\delta_1 \coloneqq  v(S \cup \{g\}) - v(S)$ and $\delta_2 \coloneqq u(S \cup \{g\}) - u(S)$. We have $\delta_1, \delta_2 \in \{0,1\}$ and 
\begin{align*}
w(S \cup \{ g\}) & = \max \{ v(S \cup \{ g \}), u(S \cup \{ g \}) \} \\
& = \max \{ v(S) + \delta_1, u(S) + \delta_2\} \\ 
& \leq \max \{ v(S), u(S) \} + \max \{ \delta_1, \delta_2 \} \tag{regrouping terms} \\ 
& \leq w(S) + 1 \tag{since $\delta_1, \delta_2 \in \{0,1\}$}
\end{align*}
Overall, these observations that the marginals of function $w$ are integers between $0$ and $1$ ensure that $w$ satisfies the binary-marginals property.  
\end{proof}

\section{Envy-Freeness and Nash Social Welfare}
\label{appendix:nomas-envy}
This section shows that, in case of binary $\XOS$ valuations, envy-free allocations can have Nash social welfare significantly lower than the optimal. Recall that an allocation $\mathcal{P} = (P_1, \ldots, P_n)$ is said to be envy free iff $v_i(P_i) \geq v_i(P_j)$ for all agents $i, j$. The example given below highlights that bounding envy between pairs of agents (or finding an arbitrary allocation that is envy-free up to one good)  does not, in and of itself, provide a constant-factor approximation for Nash social welfare.

For integer $k \in \mathbb{Z}_+$, consider a fair division instance with $n=4k$ agents and $m = 2k(2k+1)$ indivisible goods. The set of goods is partitioned into $2k+1$ disjoint subsets $M^1,M^2,\ldots, M^{2k+1}$, each with $2k$ goods. Index the goods such that $M^x = \{g^x_1, g^x_2, \ldots, g^x_{2k} \}$ for $1 \leq x \leq 2k+1$.

The first $2k$ agents have identical valuation $v(S) \coloneqq |S|$, for every $S \subseteq [m]$. The binary $\XOS$ valuation of the remaining agents $i \in \{2k+1, \ldots , 4k\}$ is as follows 
\begin{align*}
v_i(S) & \coloneqq \max_{1 \leq x \leq 2k+1} \ |S \cap M^x| \qquad \text{for every subset $S \subseteq [m]$}.
\end{align*}

We next identify an envy-free allocation $\mathcal{P} = (P_1, \ldots, P_n)$ in this instance whose Nash social welfare is $\sqrt{\nicefrac{2}{k}}$ times the optimal. Hence, by increasing $k$, one can ensure that the multiplicative gap between $\NSW(\mathcal{P})$ and the optimal Nash social welfare is arbitrarily large. 

For the first $2k$ agents $i \in [2k]$, write $P_i =  \left\{g^x_i \right\}_{x=1}^{2k}$, i.e., agent $i \in [2k]$ receives the $i$th good from every $M^x$ with $1 \leq x \leq 2k$. The $2k$ remaining goods (all in $M^{2k+1}$) are assigned as singletons to the remaining $2k$ agents: $P_j = \left\{g^{2k+1}_{j - 2k} \right\}$ for each agent $j \in \{2k+1, \ldots, 4k\}$.   The Nash social welfare $\NSW(\mathcal{P}) = \left( (2k)^{2k}\ 1^{2k} \right)^{\frac{1}{4k}} = \sqrt{2k}$.

On the other hand, the optimal Nash social welfare is at least $k$. Partition each $M^x$ into equal-sized subsets $\widehat{M}^x$ and $\widetilde{M}^x$; in particular, $\widehat{M}^x=\{g^x_1, \ldots, g^x_k\}$ and $\widetilde{M}^x=\{g^x_{k+1}, \ldots, g^x_{2k} \}$. Consider an allocation $\mathcal{N} = (N_1, \ldots, N_{4k})$ with $N_i = \widehat{M}^i$ for each $i \in [2k]$ and $N_j =\widetilde{M}^{2k-j}$ for each $j \in \{2k+1, \ldots, 4k\}$. Note that for every agent $i \in [4k]$, we have $v_i(N_i) = |N_i| = k$. Hence, the Nash social welfare of $\mathcal{N}$ is $k$. That is, the optimal Nash social welfare is at least $k$. This concludes the argument.

\end{document}